\documentclass[conference]{ieeeconf}
\IEEEoverridecommandlockouts

\usepackage{amsmath,amssymb,amsfonts}
\allowdisplaybreaks

\usepackage{amsthm}
\usepackage{algorithm,algorithmic}
\usepackage{comment}
\usepackage{mathtools}
\usepackage{interval}
\usepackage{algorithmic}
\usepackage{float}
\usepackage{graphicx}
\usepackage{textcomp}
\usepackage{xcolor}
\usepackage[sort,compress]{cite}

\usepackage{hyperref}
\usepackage{cleveref}
\crefname{lemma}{Lemma}{Lemmas}
\crefname{proposition}{Proposition}{Propositions}
\crefname{theorem}{Theorem}{Theorems}
\crefname{corollary}{Corollary}{Corollaries}
\crefname{definition}{Definition}{Definitions}
\crefname{remark}{Remark}{Remarks}

\def\BibTeX{{\rm B\kern-.05em{\sc i\kern-.025em b}\kern-.08em
    T\kern-.1667em\lower.7ex\hbox{E}\kern-.125emX}}

\theoremstyle{definition} 
\newtheorem{theorem}{Theorem}
\newtheorem{lemma}{Lemma}
\newtheorem{proposition}{Proposition}
\newtheorem{corollary}{Corollary}

\DeclareMathOperator{\diam}{diam}
\DeclareMathOperator{\Prob}{Prob}
\DeclareMathOperator{\dlyap}{dlyap}
\DeclareMathOperator{\tr}{tr}

\DeclareMathOperator{\vect}{vec}
\DeclareMathOperator{\op}{op}
\DeclareMathOperator{\sgd}{sgd}
\DeclareMathOperator{\gd}{gd}

\begin{document}

\title{\LARGE \bf On Globally Optimal Stochastic Policy Gradient Methods for Domain Randomized LQR Synthesis

}

\author{Alex Nguyen-Le and Nikolai Matni
\thanks{A. Nguyen-Le and N. Matni are with the Department of Electrical and Systems Engineering, University of Pennsylvania. Emails: \texttt{\{atn,nmatni\}@seas.upenn.edu}}%
}

\maketitle

\begin{abstract}
Domain randomization is a simple, effective, and flexible scheme for obtaining robust feedback policies aimed at reducing the sim-to-real gap due to model mismatch. While domain randomization methods have yielded impressive demonstrations in the robotics-learning literature, general and theoretically motivated principles for designing optimization schemes that effectively leverage the randomization are largely unexplored. We address this gap by considering a stochastic policy gradient descent method for the domain randomized linear-quadratic regulator synthesis problem, a situation simple enough to provide theoretical guarantees. In particular, we demonstrate that stochastic gradients obtained by repeatedly sampling new systems at each gradient step converge to global optima with appropriate hyperparameters choices, and yield better controllers with lower variability in the final controllers when compared to approaches that do not resample. Sampling is often a quick and cheap operation, so computing policy gradients with newly sampled systems at each iteration is preferable to evaluating gradients on a fixed set of systems.
\end{abstract}

\section{Introduction}
Domain randomization (DR) has gained wide adoption in reinforcement learning and robot learning for enabling sim-to-real transfer \cite{TFRSZA17}. By harnessing massive GPU parallelism, DR trains controllers across randomized environments to robustify controllers against model uncertainty by minimizing expected cost over a sampled uncertainty distribution. In contrast, control theoretic approaches typically handle uncertainty via robust control, designing for the worst case over the uncertainty set. While the robust control methods come with theoretical guarantees that are largely absent in DR, robust control methods frequently suffer from over-conservatism due to their pessimistic formulation. DR trades these theoretical guarantees for straightforward implementation and effective use of GPU computing resources, which are major contributor to DR’s practical success.

An intermediate formulation between DR and robust control, introduced by Fujinami et al.~\cite{FLMP25A, FLMP25B}, is the domain‑randomized linear‑quadratic regulator (DR‑LQR), which retains DR’s ease of implementation while adding stability guarantees. The only missing benefit of DR in Fujinami et al. \cite{FLMP25A,FLMP25B} is the effective utilization of available compute resources, which we address in our work. While Fujinami et al. can scale their policy gradient method with available GPU compute, their approach is inherently limited by the surrogate cost function employed instead of the DR cost function, which we directly optimize with stochastic gradient descent (SGD) methods.


\subsection{Related Work: LQR Policy Gradient Methods}
The seminal work in applying policy gradient methods to LQR synthesis is Fazel et al. \cite{FGKM18}, which established that the policy gradient satisfies a gradient-domination condition \cite{KNS16} that can find global optimal solutions to (LQR) synthesis problem despite nonconvexity in the cost-function and feasible set. This has been greatly refined and extended over the past decade; many other optimal control synthesis problems also support globally convergent policy optimization, notably $H_{2}/H_{\infty}$ and $H_{\infty}$ control synthesis \cite{ZHB21}, Markov jump linear system control synthesis \cite{JHD22}, and even to the reinforcement learning settings \cite{ZHB20}. Despite the impact of Fazel et al. \cite{FGKM18} in the control-literature, many extensions to LQR-synthesis problems possess convex-optimization formulations so policy-gradient methods are not the only techniques that can produce optimal controllers\cite{HZLMFB23}. A notable exception is that of the domain randomized LQR (DR-LQR) problem~\cite{FLMP25A, FLMP25B} where Fujinami et al.~\cite{FLMP25B} show that gradient descent applied to a surrogate sample-average cost converges to a nearly optimal controller.  In contrast, we show that an optimal controller can be found by applying stochastic gradient descent directly to the DR-LQR objective.


\subsection{Related Work: Domain Randomization}
Domain randomization, popularized in the reinforcement learning and robotics literature~\cite{TFRSZA17, PAZA18}, randomizes simulator parameters towards reducing the sim-to-real gap. While the parameters chosen for randomization are highly task dependent, tasks employing dynamics simulations typically vary physical constants like masses, geometries, or scaling factors \cite{PAZA18, CHMMIRF19, LKRDKS19}. Similarly, when vision is an important component for solving the task at hand, randomization is typically employed over features like textures, lighting, or field-of-view of the vision system \cite{TFRSZA17}\cite{W19}. This randomization is typically employed by sampling new parameters at each training episode, as opposed to fixed at the outset as in predecessors to DR, like the scenario-approach \cite{CGP09,TDC05} from robust control or ensemble randomization techniques from robot learning\cite{LKDRT18}\cite{MLT15}. This ability to sample and roll out randomized environments in parallel is a key factor in the success of domain randomization: we show how this capability can be leveraged to efficiently compute an optimal DR-LQR controller. 

\subsection{Main Contributions}
\begin{itemize}
    \item \textbf{Stochastic Policy Gradient Convergence:} We establish the first proof of linear convergence of minibatched SGD to a $\varepsilon$-suboptimal solution of the DR-LQR problem.  We further characterize sufficient conditions on the amount of samples a minibatch should containas a function of this suboptimality level.
    \item \textbf{Computational Advantage:} We empirically show that, for a fixed computational budget, solving the DR-LQR problem via SGD yields a policy that outperforms a controller obtained from the surrogate sample-average formulation of~\cite{FLMP25B}. 
\end{itemize}

\section{Problem Description}
Consider a fully observed linear dynamical system, with state $x_t \in \mathbb{R}^{n_x}$, input $u_t \in \mathbb{R}^{n_u}$, and noise $w_t \in \mathbb{R}^{n_x}$. The system evolves as
$$
x_{t+1} = Ax_t + Bu_t + w_t
$$
where $w_t \sim \mathcal{N}(0,\Sigma_w)$ and we assume that $\Sigma_w \succeq I$. For fixed system parameters, $\theta := [A,B]$ the infinite-horizon stochastic LQR cost is given by,
\begin{align}
J(K,\theta) &= \limsup_{T\rightarrow \infty} \mathbb{E}_w^K\left[\frac{1}{T}\sum_{t=0}^{T} \begin{bmatrix}x_t\\u_t\end{bmatrix}\begin{bmatrix}Q & S\\S^\top & R\end{bmatrix}\begin{bmatrix}x_t\\u_t\end{bmatrix}\right]\nonumber\\
&= \tr(P(K,\theta)\Sigma_w) \label{eq:lqr-cost}
\end{align}
where $P(K,\theta)$ is the cost-to-go matrix and satisfies \cite{B12},
$$
P(K,\theta) = \dlyap((A\!-\!BK)^\top, Q \!-\! SK \!-\! S^\top K^\top \!+\! K^\top RK)\\
$$
and,
$$
\dlyap(A-BK,X) = \sum_{\ell = 0}^{\infty} (A-BK)^\ell X (A-BK)^{\ell \top}.
$$
Further, the superscript $K$ indicates that the expectation is taken under the closed-loop policy $u_t=-Kx_t$ with respect to the noise process $w_t$. We also assume that the quadratic cost parameter $\mathcal{Q} \coloneqq \begin{bmatrix}Q & S\\S^\top & R\end{bmatrix}$ satisfies $\mathcal{Q} \succeq I$.

We model the system parameters $\theta$ as random variables supported on a convex and compact set $\Theta$ with density $p_\Theta$, and define the \emph{domain randomized linear-quadratic-regulator} \eqref{eq:dr_lqr} synthesis problem as seeking a feedback gain $K$ which minimizes the expecation of the LQR cost~\eqref{eq:lqr-cost} with respect to the system parameters:
\begin{equation}\label{eq:dr_lqr} \tag{DR-LQR}
\begin{aligned}
    \begin{array}{cl}
    \underset{K}{\text{minimize}} & \mathbb{E}_\theta[J(K,\theta)] \\
     \text{subject to:}& \displaystyle \sup_{\theta \in \Theta } J(K,\theta) < \infty
    \end{array}
\end{aligned}
\end{equation}
We define $J_{DR}(K) := \mathbb{E}_\theta[J(K,\theta)]$ and drop the expecation subscript, as all expectations will be with respect to $\theta$, i.e., $J_{DR}(K)=\mathbb{E}[J(K,\theta)]$.

Solving problem~\eqref{eq:dr_lqr} is challenging because both the objective and constraint set is nonconvex, and neither the objective nor its gradient can be exactly computed.  Nevertheless, we show that we can solve \eqref{eq:dr_lqr} via the minibatched SGD scheme in Algorithm \ref{alg:dr-mba}. In \S~\ref{sec:gd}, we first show that despite the nonconvexity of the underlying problem, gradient descent converges to a global optimum as the objective function enjoys a \emph{gradient dominance} property. In \S~\ref{sec:sgd}, we then show that stochastic gradient descent implemented with approximate gradients estimated via minibatch sampling also converge to a globally optimal solution, albeit at a degraded rate.

\begin{algorithm}[H]
    \caption{DR Minibatched Stochastic Gradient Descent} 
    \label{alg:dr-mba}
    \begin{algorithmic}[1]
    \STATE \textbf{Input:} System Probability Distribution $p_\Theta(\cdot)$, Minibatch Size $M$, Step-Size, $\eta$, Number of Gradient Steps $t$, Starting Point $K_0$
    \STATE \textbf{Output:} K
    \STATE $K \gets K_0$
    \FOR {$n \in \{1,\ldots, N\}$}
    \STATE Sample $M$ systems from $p_\Theta$
    \vspace{-0.5\baselineskip}
    $$
        \theta_1,\ldots,\theta_M \sim p_{\Theta}(\theta)
    $$
    \vspace{-1.25\baselineskip}
    \STATE Compute the minibatched gradients
    \vspace{-.5\baselineskip}
    \begin{equation}\label{eq:MB-G}\tag{MB-G}
        g(K) \gets \frac{1}{M}\sum_{i = 1}^M \nabla J(K,\theta_i)
    \end{equation}
    \vspace{-0.75\baselineskip}
    \STATE Take a minibatched gradient step,
    \vspace{-.5\baselineskip}
    \begin{equation} \label{eq:GD}\tag{MB-SGD}
        K \leftarrow K - \eta g(K) 
    \end{equation}
    \ENDFOR
    \end{algorithmic}
\end{algorithm}
Proofs of all intermediate results, as well as exact expressions for constants and polynomials appearing in the sequel can be found in the appendix of our extended version available.
\section{Convergence Analysis of Gradient Descent for DR-LQR}
\label{sec:gd}


We begin with the idealized setting in which the exact gradient $\nabla J_{DR}(K)$ of the DR-LQR objective is available. Even in this case, establishing convergence of gradient descent to a global optimum is nontrivial since the DR-LQR cost function is nonconvex in the feedback gain $K$. This difficulty already arises in the LQR problem parameterized by the control policy~\cite{FGKM18}, where global convergence of policy gradient methods is shown by exploiting a \emph{gradient dominance property} and suitable smoothness conditions. In particular, in addition to establishing gradience dominance, smoothness is needed to ensure that each update  
\[
K_{\ell+1} = K_\ell - \eta \nabla J_{DR}(K)
\]  
remains \emph{stabilizing}, i.e., in the feasible set of gains. 

In the standard LQR setting,  Hu et al.~\cite{HZLMFB23} address this challenge through a modified smoothness argument: by restricting attention to line segments between feasible policies, they establish a form of local smoothness that guarantees cost decrease under an appropriately chosen step size. Here, we extend this reasoning to the DR-LQR setting by showing that the DR-LQR cost function is coercive, $L$-smooth (i.e., has $L$-Lipschitz gradients) on any sublevel set, and satisfies a (local) gradient dominance property near the optimal solution. To show convergence, we also assume that an initial controller $K_0$ is available which simultaneously stabilizes all systems in $\Theta$, i.e, that $K_0 \in \mathcal{K}_\Theta:=\{K\, : \, \sup_{\theta\in\Theta} J(K,\theta)<\infty\}$. 




\subsection{DR-LQR Cost is Coercive} We note that for any $\theta\in\Theta$, the LQR cost $J(K,\theta)$ is coercive over any sub-level set $\{K \, : \, J(K,\theta)\leq c\}$ for $c<\infty$~\cite[Lemma 3.7]{BMFM19}.  A minor modification of Lemma 3.7 in \cite{BMFM19} shows that $J(K,\theta)$ satisfies the uniform in $\theta$ quadratic lower bound $J(K,\theta)\geq \lambda_{\min}(\Sigma_w)\|K\|_F^2$.  Taking expectations of both sides then allows us to conclude coercivity of the DR-LQR cost $J_{DR}(K).$ 

\subsection{ DR-LQR Cost is Locally L-Smooth} 

%
The cost $J_{DR}(K)$ is not globally $L$-smooth because for any fixed $\theta\in\Theta$, the cost $J(K,\theta)$ rapidly becomes infinite at the boundary between stable and unstable controllers, violating traditional smoothness conditions~\cite{HZLMFB23}: hence we do not expect the DR-LQR cost $J_{DR}(K)$ to satisfy global smoothness either. Instead, we fix a performance level $c>o$, and show that within $c$-sublevel sets
$S_c \coloneqq \{ K : J_{DR}(K) \leq c\},$ local smoothness can be established.

\begin{lemma}[Policy-Smoothness]\label{lem:policy_smoothness} 
Fix $c>0$. The cost $\mathbb{E}[J(K,\theta)]$ is $L_K$-smooth at any $K\in S_c$, with $L_K$ given by
\begin{align}
    L_K:=&4\left(\left\Vert\mathcal{Q}\right\Vert_{\mathrm{op}} + \frac{2\big\Vert\bar{\theta} \big\Vert_{\mathrm{op}}^2 J_{DR}(K_0)}{\sigma_{\mathrm{min}}(\Sigma_w)} \right)\cdots \label{eq:LK_upper_bound}\\
    &\qquad\cdot\left(1 + \frac{4\Vert \bar{\theta}\big\Vert_{\mathrm{op}}^2J_{DR}(K_0)}{\sigma_{\mathrm{min}}(\Sigma_w)^2} \right)\frac{2J_{DR}(K_0)}{\sigma_{\mathrm{min}}(\mathcal{Q}) }\nonumber
\end{align}
Where the operator norm is given by the maximum singular value. Here,
\begin{align*}
    E(K,\theta) &= (R\!+\! B^\top P(K,\theta)B)K \!-\! B^\top P(K,\theta)A \!-\! S^\top\\
    \Sigma(K,\theta) &= \dlyap((A-BK)^\top,\Sigma_w)\\
    \bar{\theta} &= \sup_{\theta\in\Theta} \Vert [ A,B]\Vert_{\op}.
\end{align*}
We note that all quantities in \cref{eq:LK_upper_bound} can be bounded above by polynomials in problem specific parameters and $c$.
    \begin{proof}
        The Hessian of $J(K,\theta)$ is given in Bu et al. \cite{BMFM19}. Employing this Hessian, submultiplicativity, subadditivity, the shorthands introduced, and finally
        noting that, by Appendix \cref{cor:hessian_lipschitz_characterization}, $\|\mathbb{E}[D_K^2J(K,\theta)]\|_{\text{op}}\leq L_K$.
    \end{proof}
\end{lemma}

We next show that local $L$-smoothness can be used to select a step-size for gradient descent that ensures that the DR-LQR cost decreases at each iteration.  The following is an adaptation of Hu et al. \cite[Theorem 1]{HZLMFB23} to obtain the following feasibility lemma.

\begin{proposition}[Theorem 1 of Hu et al.~\cite{HZLMFB23}] \label{prop:implicit_regularization} Let $K_1,K_2 \in S_c$ be two arbitrary policies in the sublevel set $S_c$. Denote by
$\mathcal{I}(K_1,K_2) = \{\phi K_1 + (1-\phi)K_2: \phi \in [0,1]\},$
the line segment connecting $K_1,K_2$. If $\mathcal{I}(K_1,K_2) \subseteq S_c $, then for $L_K$ as given in equation~\eqref{eq:LK_upper_bound}, we have that
\begin{align}
J_{DR}(K_2) &\leq J_{DR}(K_1) + \langle \nabla J_{DR}(K_1), K_2-K_1\rangle\notag\\
&\qquad +\frac{L_K}{2}\Vert K_2-K_1\Vert_F^2 \label{eq:descent_lemma}
\end{align}
\end{proposition}

We can apply this result to show that if $K_0\in S_c$, then all future iterates $K_\ell$ of gradient descent applied with step-size $\eta=L_K^{-1}$ remain in $S_c$. In particular, it holds that for any $K\in S_c$, the segment $\mathcal{I}(K,K-L_K^{-1}\nabla \mathbb{E}[J(K,\theta)]) \subseteq S_c$, so that by Prop.~\ref{prop:implicit_regularization}:
\begin{align*}
    &J_{DR}(K - L_K^{-1}\nabla_K J_{DR}(K)) \leq\\
    &\qquad\qquad J_{DR}(K) - \frac{1}{2L_K} \Vert \nabla_K J_{DR}(K)\Vert_F^2 \leq J_{DR}(K).
\end{align*}

\vspace{0.45\baselineskip}
This is called the \emph{implicit regularization property} by Hu et al. \cite{HZLMFB23}, and importantly ensures that all iterates of gradient descent remain in $\mathcal{K}_\Theta$ if the initial iterate $K_0\in\mathcal{K}_\Theta.$

\subsection{Gradient Dominance of DR-LQR Cost:} 
Fujinami et al.~\cite{FLMP25B} show that a sample average surrogate of the DR-LQR cost enjoys gradient dominance in a neighborhood of the optimal policy under suitable heterogeneity assumptions, as captured by $\diam(\Theta)$.


We show here that this argument can be extended to the true DR-LQR cost function $J_{DR}(K)$ under analogous assumptions.  Our strategy is to first show that the DR-LQR cost satisfies an approximate gradient dominance property, then argue that this guarantees gradient dominance in a suitable neighborhood of the optimal controller $K^\star.$

\begin{lemma}[Approximate Gradient-Domination] Fix $c>0$ and $K\in S_c$. Then the following bound holds:
\begin{align*}
&J_{DR}(K)-J_{DR}(K^\star) \\
&\qquad\leq 2 \left(\Vert  \nabla_K J_{DR}(K) \Vert_F^2 + 2\Vert\mathbb{E}[ R(K,K^\star,\theta)]  \Vert_F^2\right)
\end{align*}
where
    $$
    R(K^\star,K,\theta) = E(K,\theta)(\Sigma(K,\theta)- \Sigma(K^\star,\theta)),
    $$
    \begin{proof}
    This proof is an adapation of Fujinami et al. Lemma III.3 \cite{FLMP25B}. By the almost-smoothness lemma of Fazel et al. \cite{FGKM18}, combined with the monotonic property of the integrals, we have that:
    \begin{align*}
        &J_{DR}(K)-J_{DR}(K^\star) \leq \\
        & 2\vect(K-K^\star)^\top \vect\big(\mathbb{E}[ \Sigma(K^\star,\theta)E(K,\theta))]\big)\\
        &-\vect\big(K-K^\star)^\top \mathbb{E}[ \big(\Sigma(K^\star,\theta) \otimes (B^\top P(K,\theta)B + R)](\cdot)\\
        &\leq 2 \left(\Vert  \nabla_K J_{DR}(K) \Vert_F^2 + 2\Vert\mathbb{E}[ R(K,K^\star,\theta)] \Vert_F^2]\right).
    \end{align*}
    Where $(\cdot)$ is $\vect(K-K^\star)$; the result follows by employing completion of squares and adding/subtracting $\Sigma(K,\theta)$.
    \end{proof}
\end{lemma}

From here, we see that if we can bound the remainder term $R(K,K^\star)$ by a sufficiently small multiple of $\mathbb{E}[J(K,\theta)-J(K^\star,\theta)]$, we can ensure gradient domination---this is the approach taken in Lemma III.4 of Fujinami et al. \cite{FLMP25B}. We do so in the following Lemma, which extends the arguments of~\cite[Lemmas III.5-7]{FLMP25B} to the DR-LQR cost $J_{DR}$.
\begin{proposition}\label{prop:heterogeneous_sublevel_set}
Let the diameter of the set be bounded above by,
$$
\diam(\Theta) \leq \inf_{\theta \in \Theta} \frac{1}{50000 \max(\Vert \bar{\theta} \Vert,1) J(K_{\mathrm{LQR}}(\theta),\theta)^6 }
$$
where $\Theta_B$ is the support of $B$. Within the set,
$$
\mathcal{S} \!\coloneqq\! \left\{ \begin{array}{c|c} \!\!\!\!K\! &\begin{matrix}J_{DR}(K)\leq8J_{DR}(K^\star),\\
\!\Vert \nabla \displaystyle J_{DR}(K) \Vert_F \! \leq \! \frac{1}{2^8 \! \displaystyle\sup_{B\in\Theta}(\Vert B \Vert,1)J_{DR}(K)^3}\end{matrix} \end{array}\!\!\!\!\right\},
$$
we have that
\begin{align*}
J_{DR}(K)-J_{DR}(K^\star) &\leq 4\Vert  \nabla J_{DR}(K)\Vert_F^2.
\end{align*}
\end{proposition}


\noindent\textbf{Convergence of Gradient Descent:} 

We now have all the ingredients needed to prove convergence of gradient descent to an optimal solution to problem~\eqref{eq:dr_lqr} assuming an initial stabilizing controller $K_0\in\mathcal{K}_\Theta\cap S$.\footnote{We show in the appendix that under gradient dominance, a slight modification of the discounted-annealing scheme proposed in Fujinami et al. \cite{FLMP25B} enables one to find such an initial controller $K_0\in\mathcal{K}_\Theta\cap\mathcal{S}$.}  

\begin{theorem}[Linear Convergence]\label{thm:risk_gradient_convergence}
    Let $K_0\in\mathcal{K}_\Theta\cap \mathcal{S}$, set $\eta\leq L_K^{-1}$, and consider the gradient descent iterates
    $$
    K_{\ell+1} = K - \eta \nabla J_{DR}(K).
    $$
    Then, the following are true:
    \begin{itemize}
        \item For all $\ell=0,1,2,...$, we have that $K_\ell\in\mathcal{K}_\Theta\cap S.$
        \item The iterates satisfy 
        \begin{align}
    &J_{DR}(K_\ell) - J_{DR}(K^\star) \nonumber\\
    &\leq \left(1 -  \frac{1}{8L_K}\right)^\ell\big(J_{DR}(K_0) - J_{DR}(K^\star)\big). \label{eq:linear_rate_exact_gradient}
    \end{align}
    \end{itemize}
    In particular, if we take
    \begin{equation} \label{eq:num_grad_steps} N\geq \log\Big(\frac{J_{DR}(K_0) - J_{DR}(K^\star)}{\varepsilon }\Big)\Big/\log\Big(\frac{8L_k}{8L_K-1}\Big)
    \end{equation}
    gradient descent steps, we achieve suboptimality,
    $$
    J_{DR}(K_N) - J_{DR}(K^\star) \leq \varepsilon.
    $$
    \begin{proof}
        The result follows by~\cite[Theorem 1]{HZLMFB23}.  In particular, we have proved that for any $K\in \mathcal{K}_\Theta\cap S$, the DR-LQR cost $J_{DR}(K)$ is coercive and is $L_K$-smooth (Lem.~\ref{lem:policy_smoothness}).  Further note that if $K_0\in \mathcal{K}_\Theta\cap S$ then all future iterates $K_\ell$ also belong to the set $\mathcal{K}_\Theta\cap S$ follows from the descent guarantee following Prop.~\ref{prop:implicit_regularization} and the definition of $\mathcal{S}$, as gradient norms of iterates are bounded by a monotonically decreasing function.  Finally, as the DR-LQR cost function satisfies a gradient dominance property of degree 2 with parameter $\mu=\frac{1}{8}$, the linear convergence rate~\eqref{eq:linear_rate_exact_gradient} is guaranteed.
    \end{proof}
\end{theorem}
\section{Convergence Analysis of SGD for DR-LQR}
\label{sec:sgd}
This section provides an end-to-end convergence analysis for the minibatch SGD procedure in \Cref{alg:dr-mba}, where the primary goal is find hyperparameters guaranteeing a $\varepsilon$-suboptimal solution with probability at least $1-\delta$. This requires explicit characterization of a minibatch size $M$, step size $\eta$, and number of descent steps, $N$. The main technical challenge is the potential for infeasibility arising from gradient approximation. While we know from \cref{prop:implicit_regularization} that gradient steps maintain feasibility, it is possible for approximate gradient steps to fail to maintain feasibility due to  gradient-approximation errors.

\subsection{SGD Analysis Setup}
Throughout this section we assume that $J_{DR}(K)$, is $L_K$-smooth on its sublevel-sets as previously established in \Cref{sec:gd}, and $\mu$-gradient dominated on $\mathcal{S}$ (as in \cref{prop:heterogeneous_sublevel_set}), and on a fixed sublevel set $S_c$. Further, we define, the minibatch gradient $g(K)$, exact gradient step $K^{\gd}$, and stochastic gradient step $K^{\sgd}$ as
\begin{align*}
    g(K) &\coloneqq \frac{1}{M}\sum_{i=1}^{M}\nabla_K J(K,\theta_i), \qquad \theta_i \overset{\text{iid}}{\sim} p_\Theta \tag{MB-G}\\
    K^{\gd} &\coloneqq K - L_K^{-1} \nabla J_{DR}(K)\\
    K^{\sgd} &\coloneqq K - L_K^{-1}  g(K).
\end{align*}
This section analyzes the convergence properties of the SGD iterates:
\begin{equation}\label{eq:sgd_iterates}
K_{\ell+1}  = K_{\ell } - L_K^{-1} \nabla g(K_\ell) = K^{\sgd}_{\ell} . \tag{MB-SGD}
\end{equation}
To analyze how accurately $g(K)$ must estimate the gradients, it is helpful to interpret the stochastic gradient update as a perturbation to the gradient update,
$$
K^{\sgd}_{\ell} = K_{\ell}^{\gd} + (K^{\sgd}_{\ell}- K_{\ell}^{\gd}).
$$
Employing this with known feasibility guarantees on $J(K,\theta)$, like the ones in Fazel et al. \cite{FGKM18}, allows us to derive similar guarantees for $J_{DR}(K)$. We condense and summarize the contents of lemma 22 through 27, relegating its proof to the appendix \cite{FGKM18}
.
\begin{proposition}[DR Feasible Step Size]{\label{prop:feasible_step_size}}
Consider any $K$ belonging to a sublevel set $ S_c$ of $J_{DR}$. If
$$
\Vert K^{\sgd} - K^{\gd}\Vert_{\op} \leq \frac{\sigma_{\mathrm{min}}(Q) \sigma_{\mathrm{min}}(\Sigma_w)}{4 J_{\mathrm{DR}}(K) \big\Vert \bar{\theta}\big\Vert_{\mathrm{op}}\big( \big\Vert \bar{\theta}\big\Vert_{\mathrm{op}}  + 1\big)} \eqqcolon c_g
$$
then
$$
|J_{DR}(K^{\gd}) - J_{DR}(K^{\sgd})| \leq L_{cost}\Vert K^{\gd} - K^{\sgd}\Vert_{F}.
$$
Furthermore, the entire line segment, $\mathcal{I}(K,K^{\sgd} )$ is feasible.
\end{proposition}

 The derivations of the constant $c_g$ and $L_{cost}$, follow by application of matrix calculus identities are provided in the appendix of the extended version of this work. \Cref{prop:feasible_step_size} ensures that we possess enough smoothness for gradient descent to be tunable and yield feasible iterates, a key step in invoking gradient-domination bounds.
\subsection{One-Step Decomposition}
To analyze the iterates obtained by SGD \cref{eq:sgd_iterates}, we first decompose the cost at each iteration into an exact gradient update step perturbed by an error due to sampling.
\begin{lemma}[Cost Decomposition, \cite{FGKM18} Thm. 31] \label{lem:gradient_step_cost_decomposition} Assume that $\Vert K^{\gd} - K^{\sgd}\Vert$ satisfies \cref{prop:feasible_step_size}. The cost difference can be bounded above by,
\begin{align}
J_{DR}(K^{\sgd}) \!-\! J_{DR}(K^{\star}) &\!\leq\! \Big(1 \!-\! \frac{\mu}{L_K} \Big) \big(J_{DR}(K) \!-\! J_{DR}(K^\star)\big)\notag\\
& \qquad + L_{\mathrm{cost}}\Vert K^{\sgd}- K^{\gd}\Vert_F \label{eq:grad_step_cost_decomp}
\end{align}
\begin{proof}
    By \cref{prop:implicit_regularization}, we have that $K^\mathrm{gd}$ is feasible, and thus satisfies:
    \begin{align}
        &J_{DR}(K^{\gd}) \!-\! J_{DR}(K^\star)\!\leq\! \left(\!\!1 - \frac{\mu}{L_K}\!\!\right)(J_{DR}(K) - J_{DR}(K^\star))\nonumber
    \end{align}
    Adding and subtracting $J_{DR}(K^{\sgd})$ to the left-hand side, rearranging the terms, and applying \cref{prop:feasible_step_size} yields the result.
\end{proof}
\end{lemma}
This one-step decomposition reveals the terms that need to be controlled to establish convergence of SGD. In the upper bound of \cref{eq:grad_step_cost_decomp}, we have a contraction term of \cref{eq:linear_rate_exact_gradient}, perturbed by an error term due to gradient approximation. We use a sufficient condition for descent adapted from Fazel et al. \cite{FGKM18}, which accounts for the error as a contraction rate degradation.

\begin{lemma}[Degraded Contraction]\label{lem:degraded_contraction} Fix a desired suboptimality $\varepsilon$. For $\ell \in \{0,\ldots,N-1\}$, if we have
\begin{align*}
\Vert K_\ell^{\sgd}- K_\ell^{\gd}\Vert_F&\leq \frac{\mu\varepsilon}{2L_{\mathrm{cost}}L_K}, \qquad \varepsilon \leq J_{DR}(K_\ell) - J_{DR},
\end{align*}
we obtain the degraded contraction rate 
$$
J_{DR}(K_\ell) - J_{DR}(K^\star) \!\leq\! \Big(\!1\!-\!\frac{\mu}{2L_K}\!\Big)^\ell (J_{DR}(K_0) - J_{DR}(K^\star) ).
$$
\begin{proof}
This follows by an induction, following \cite[Theorem 31]{FGKM18}, \cref{lem:gradient_step_cost_decomposition}, and the additional assumptions we make.
    \begin{corollary}[Gradient Estimation Error]\label{cor:allowable_estimation_error} If, at each iteration, we have that
    $$
    \Vert g(K_\ell) - \nabla J_{DR}(K_\ell) \Vert_F \leq \frac{\mu\varepsilon}{2L_{\mathrm{cost}}},
    $$
    then we obtain the degraded contraction rate of \cref{lem:degraded_contraction}.
    \end{corollary}
\end{proof}
\end{lemma}

A direct consequence of \cref{cor:allowable_estimation_error}, is a link between allowable gradient estimation errors to and the corresponding policy errors. We can arrive at an expression for the required accuracy of the minibatch gradients to remain feasible, and yield a degraded linear convergence rate.
All of these quantities can be expressed in terms of problem specific constants. Finally, we address the issue of minibatch-gradient concentration.
\subsection{Concentration of minibatched-gradient estimators}
We start by denoting the gradient estimation error by $\varepsilon_\nabla$, where we would like to ensure that:
$$
\Vert g(K)- \nabla J_{DR}(K) \Vert_F \leq \varepsilon_{\nabla}.
$$
for some choice of gradient estimation error $\varepsilon_\nabla$.
We achieve this by invoking the Matrix Bernstein inequalities, which require that we find norm and variance bounds on \eqref{eq:MB-G}. This can be achieved by defining the zero-mean symmetrized matrix,
\begin{align}
    G_M &\coloneqq \begin{bmatrix} 0 & \vect\big(g(K) - \nabla_K J_{DR}(K)\big)\\ * & 0 \end{bmatrix}= \sum_{i=1}^M G_i\label{eq:sym_bernstein}
\end{align}
that we will bound with the rectangular matrix Bernstein inequality, found in Vershynin \cite[Exercise 5.4.15]{V18}, adapted to the task at hand. Bounds on the support norm $\bar{G}$ and variance proxy $\sigma^2$ as expressions in the problem constants are characterized in our appendix. Essentially, we employ the mean-value theorem to obtain these in terms of problem specific parameters and $\diam(\Theta)$.

\begin{figure*}[!t]
    \centering
    \includegraphics[width=\linewidth]{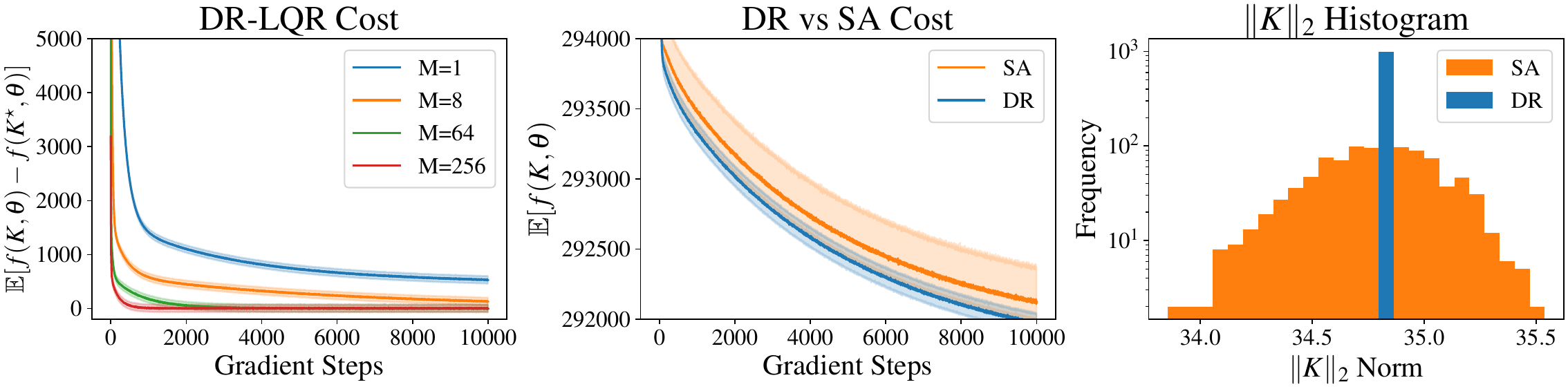}
    \caption{All figures are generated with 1000 independent trials with 10000 gradient steps. \textbf{Left:} We plot the median and the 25\%  and 75\% percentiles for 1000 independent trials of optimizing the domain randomization cost with varying minibatch size. To approximate the DR-cost we employ the sample average estimator,
    $
    J_{DR}(K) \approx \frac{1}{n_\text{dr}}\sum_{i=1}^{n_\text{dr}} J(K,\theta_i)
    $
    with a large amount of samples ($n_{dr} = 10^5$) used to visualize the descent dynamics of the DR-cost. \textbf{Center:} We provide a zoomed in figure for the descent dynamics of the SA synthesized controller of \cite{FLMP25B} and the DR synthesized controller by \cref{alg:dr-mba} in this work with $M=8$ for both optimization procedures, as well as the 25\% and 75\% percentiles of the cost trajectories. \textbf{Right:} We plot the empirical distribution (on a logarithmic scale) of the $\ell_2$ norm of the final obtained controller to illustrate the reduced variance in the controller synthesized by DR versus SA.}
    \label{fig:DR-LQR}
\end{figure*}
\begin{lemma}[Matrix Bernstein Inequality]\label{lem:matrix_bernstein} Let $G_i$ be defined as in \eqref{eq:sym_bernstein}, and assume that,
$$
\begin{gathered}
\Vert\nabla_K J(K,\theta_i) - \nabla J_{DR}(K)\Vert_F \leq \bar{G}\\
\mathbb{E}[\Vert\nabla_K J(K,\theta_i) - \nabla_KJ_{DR}(K)\Vert_F^2]   \leq \sigma^2
\end{gathered}
$$
then,
\begin{align*}
\Prob\left\{ \Vert G_M\Vert_{F} \geq \varepsilon_{\nabla}  \right\} \!\leq \!2(n_xn_u\!+\!1)\!\exp\!\left(\!\frac{-M \varepsilon_\nabla^2}{4 (\sigma^2 \!+\! \bar{G}\varepsilon_\nabla/3) }\!\right)\!.
\end{align*}
Thus, if we pick $M$ as,
$$
M \geq \frac{4(\sigma^2 + \bar{G}\varepsilon_\nabla/3)}{\varepsilon_{\nabla}^2}\log\left(\frac{2(n_xn_u+1)}{\delta}\right),
$$
with probability at least $1-\delta$, each gradient-estimation error is bounded above by $\varepsilon_\nabla$.
\begin{proof}
The matrix Bernstein inequality we use is expressed as a Frobenius norm; the conversion and expression in problem dependent constants is covered in Appendix \cref{lem:var_poly_bound}.

\end{proof}
\end{lemma}
\subsection{Convergence of Stochastic Gradient Descent}
Finally, we have all of the expressions needed to find an explicit set of hyperparameters to make \cref{alg:dr-mba} converge to an $\varepsilon$-suboptimal policy,  expressed solely in terms of the problem specifications.
\begin{theorem}[High-Probability Convergence of Minibatch SGD]\label{thm:hyperparamater_tuning} Consider \eqref{eq:dr_lqr} problem where the conditions of \cref{prop:heterogeneous_sublevel_set} are satisfied. Then, we have that $J_{DR}(K)$ is $\mu$-gradient dominated with $\mu = 1/8$ and with access to a feasible starting controller $K_0 \in \mathcal{S}$. Pick a desired suboptimality, $\varepsilon
$ and failure probability $\delta$. If we run \cref{alg:dr-mba} with:
\begin{align*}
    \eta &\leq L_K^{-1} \nonumber\\
    N &\geq \log\Big(\frac{J_{DR}(K_0) - J_{DR}(K^\star)}{\varepsilon }\Big)\Big/\log\Big(\frac{8L_k}{8L_K-1}\Big)\nonumber\\
    M &\geq \frac{4(\sqrt{2}\sigma^2 + \bar{G}\varepsilon_\nabla/3)}{\varepsilon_{\nabla}^2}\log\left(\frac{2N(n_xn_u + 1)}{\delta}\right), \label{eq:batch_size}\\
    \varepsilon_\nabla &\leq \min\left\{\frac{\mu\varepsilon}{2L_{\mathrm{cost}}}, c_g\right\}
\end{align*}
then, with probability at least $1-\delta$ over the randomness of the minibatch sampling,
$$
J_{DR}(K) - J_{DR}(K^\star) \leq \varepsilon.
$$
Where all constants depend only on $\mu$, $K_0$, $\mathcal{Q}$, $\Sigma_w$, $n_x$, $n_u$, $\diam(\Theta)$, $\Vert \bar{\theta} \Vert$, $\varepsilon$, and $\delta$.
\begin{proof} We already have expressions for the step size and the number of iterations, (that can also account for the degraded contraction rate) from \cref{eq:LK_upper_bound} and \cref{eq:num_grad_steps},
where we use the degraded contraction rate of \cref{lem:degraded_contraction}. However, to find a minibatch size, we need to ensure that all $N$ iterations to succeed. If we specify a probability of success $\delta$, and set the failure probability at each stage to be $\delta/N$, then $M$ is given by \cref{lem:matrix_bernstein}. Finally, to find $\varepsilon_{\nabla}, \bar{G}, \sigma^2$, we employ \cref{cor:allowable_estimation_error} and \cref{lem:gradient_step_cost_decomposition} to express $\varepsilon_\nabla$ in terms of problem dependent parameters and a starting point, and \cref{lem:var_poly_bound} to express $\bar{G}$ and $\sigma^2$ in terms of in terms of problem dependent parameters and a starting point.
\end{proof}
\end{theorem}
We emphasize that with this batch size $M$, the intermediate iterates are also feasible with high-probability, and that  \Cref{alg:dr-mba} with parameters prescribed by \cref{thm:hyperparamater_tuning} only yields a $\varepsilon$-suboptimal controller, but taking any additional steps is not guaranteed to yield any further progress. Refinements with step size tunings that yield globally convergent SGD we leave for future work. 

\section{Numerical Experiments}

We demonstrate that we can solve the DR-LQR problem defined in equation~\eqref{eq:dr_lqr} in practice with minibatched SGD. In our experiments, we consider stabilizing a cart-pole system linearized by automatic differentiation and discretized using the matched transformation. For our randomization domain, $\Theta$, we randomize pendulum lengths that we assume come from a uniform distribution supported on $[0.2,0.8]\,\mathrm{m}$, a cart mass of $1$kg, a pendulum mass of 1kg, a track-cart coefficient of friction $0.25$, and a pendulum joint friction $0.01\, \mathrm{N}\cdot m$. To generate our figures, we plot data from 1000 independent trials, with step-size $\eta = 5\times 10^{-8}$ and number of gradient steps $t = 10000$. We do not use \cref{thm:hyperparamater_tuning} directly to find $M$, $N$ or $\eta$, instead finding them empirically by hyperparameter tuning. Our conservative step size was tuned specifically to ensure convergence on all batch sizes for equal comparisons of optimizer progress versus number of gradient steps. All other implementation details will be relegated to the appendix.

Figure \ref{fig:DR-LQR} (left) illustrates that that increasing minibatch size leads to more accurate estimates of the gradient which leads to faster convergence of minibatched SGD. Further, more accurate gradients lead to better estimation of the global optima due to the lower gradient error incurred by sampling, as predicted by standard SGD theory and tighter concentration of the minibatch gradient about $\nabla J_{DR}(K)$\cite{GG23}.

We also compare this to the sample-averaged (SA) approximation to the DR cost, which first fixes $M=8$ systems to form a surrogate cost function to optimize with as in \cite{FLMP25B}. This is illustrates the case of limited memory, where the optimizer employed only has enough memory to store $M$ systems, where differences between DR and SA are more substantial. We observe that the controllers found by DR achieve lower costs, and have lower variability compared to their SA counterparts in the sense that the $25^\mathrm{th}$ and $75^\mathrm{th}$ percentiles are closer together. Further, the right-figure plotting the $\ell_2$ norm of the eventual controller shows that in the parameter space, there is substantially lower variation in controllers found by DR when compared to SA. Finally, we emphasize that while the stochasticity in the DR cost in \ref{fig:DR-LQR} (center) suggest that the controllers synthesized by DR can vary significantly in cost, this stems from the DR cost approximation, not from inconsistency in the final converged controller.

In our experiments, the SA optimization procedure of Fujinami et al. \cite{FLMP25B} takes 33.1 s (on average) to synthesize a controller for a fixed set of $M=8$ systems, while the DR optimization routine takes 40.5 s (on average) to run on a Nvidia A100 GPU. For a modest 22\% increase in runtime, we obtain controllers with lower DR cost and lower $\Vert K\Vert_{2}$ variance. This demonstrates a clear and practical advantage of the DR formulation; for a marginal increase in computational effort, the algorithm produces controllers that are not only higher-performing but, critically, are also substantially more reliable and consistent.
\section{Conclusions and Future Work}
We have studied a stochastic policy gradient method for the domain randomized linear-quadratic regulator problem. We showed that repeatedly sampling new systems at each gradient step led to convergence to global optima under appropriate hyperparameter choices. Compared to approaches that rely on a fixed set of systems, our method produced better controllers with lower variability in the final performance. Since sampling is often computationally inexpensive, evaluating gradients on freshly sampled systems proves to be more effective than using a static collection.

Our future directions are primarily targeted as loosening the heterogeneity conditions employed to establish gradient-dominance, as the ones employed here are very restrictive. Further, we also believe it is possible to strengthen our methods to yield globally convergent SGD instead of just $\varepsilon$-suboptimality. As this is the case in analyses of SGD under conditions of gradient dominance, we suspect that the extension to settings with nonconvex feasible sets in accessible with the tools developed here and more broadly in the optimization literature.

\section{Acknowledgements}
A. Nguyen-Le thanks Bruce Lee and Tesshu Fujinami for their insightful comments. This work was supported in part by NSF CAREER award ECCS-204583.

\bibliographystyle{IEEEtran}
\bibliography{IEEEabrv,references}
\onecolumn
\newpage
\section{Appendix}
    \subsection{Gradient Domination}
    We step through the main points of the proof to establish that $J_{\text{DR}}$ is gradient dominated, which starts by establishing that under the prescribed homogeneity, we have that for all feasible $K$,
    $$
    J_{\mathrm{DR}}(K) \leq 8J_{\mathrm{DR}}(K^\star) \qquad \implies \qquad J_{\text{LQR}}(K,\theta) \leq 2J_{\mathrm{DR}}(K)
    $$
    From here, we recall that with the substitution $\begin{bmatrix}\bar{A} & \bar{B}\end{bmatrix} = \Sigma_w^{1/2}\begin{bmatrix}A & B\end{bmatrix}$, this allows us to consider noise covariance $I$ to get $J_{DR}(K) = \mathbb{E}[\mathrm{tr}(\tilde{P}(K,\theta))]$, where $\bar{P}(K,\theta) = \mathrm{dlyap}(\bar{A} - \bar{B}K,I)$, where we will abuse notation with $\bar{P}(K,\bar{\theta}) = P(K,\theta)$.
    \begin{lemma}[$(\mathbf{B},s)$-boundedness, \cite{FLMP25B} Lemma III.5] \label{lem:bs_boundedness}
        Under the assumption that:
        $$
        \diam(\Theta) \leq \inf_{\theta \in \Theta} \frac{1}{50000 \displaystyle\sup_{B\in \Theta_B}(\Vert B \Vert,1) J(K_{\mathrm{LQR}}(\theta),\theta)^6 }
        $$
        we have that:
        $$
        J_{\mathrm{DR}}(K) \leq 8J_{\mathrm{DR}}(K^\star) \qquad \implies \qquad J_{\text{LQR}}(K,\theta) \leq 2J_{\mathrm{DR}}(K)
        $$
        where $B = 8$ and $s = 2$.
        \begin{proof}
        Following the strategy in \cite{FLMP25A} Lemma III.5, we begin by:
            \begin{align*}
                J_{\text{LQR}}(K,\tilde{\theta}) &\coloneqq \min_{\theta \in \Theta} J_{\text{LQR}}(K,\theta) \leq \underset{\theta}{\mathbb{E}} [J_{\text{LQR}}(K,\theta)] \leq 8\underset{\theta}{\mathbb{E}}[J_{\text{LQR}}(K_{\text{DR}}^\star, \theta)] \leq 8\underset{\theta}{\mathbb{E}}[3J_{\text{LQR}}(K(\theta),\theta)] \leq 24 J_{\text{DR}}(K,\theta)\\
                \varepsilon_{\mathrm{het}} &\leq \frac{1}{50000\cdot\displaystyle \sup_{\theta \in \Theta}\tau_B(\theta) J_{DR}(K)^6}\leq \frac{1}{16(24)^2J_{DR}(K)^2}\leq \frac{1}{16\big(24J_{DR}(K)\big)^2}\leq \frac{1}{16J(K,\tilde{\theta})^2}
    \end{align*}
    Again, following \cite{FLMP25A} with $\mathrm{tr}(P(K,\theta) - P(K,\tilde{\theta})) \leq 8\mathrm{tr}\big( P(K,\theta) \big)\big\Vert P(K,\tilde{\theta}) \big\Vert_{\op}^2 \varepsilon_{\mathrm{het}} \leq \frac{1}{2}\mathrm{tr}\big( P(K,\theta) \big)$
    \begin{align*}
    J_{\mathrm{LQR}}(K,\theta) - J_{\mathrm{LQR}}(K,\tilde{\theta}) &\leq 8 J_{\mathrm{LQR}}(K,\theta) \Vert P(K,\tilde{\theta})\Vert^2\varepsilon_{\mathrm{het}}\leq \frac{1}{2}J_{\mathrm{LQR}}(K,\theta)\\
    \intertext{where the last inequality follows by plugging in the expression for $\varepsilon_{\mathrm{het}}$, which implies:}
    J_{\mathrm{LQR}}(K,\theta) &\leq 2J_{\mathrm{LQR}}(K,\tilde{\theta}) \leq 2J_{\mathrm{DR}}(K)
    \end{align*}
        \end{proof}
    \end{lemma}
    \noindent This is the first step in establishing gradient dominance, the major accomplishment in Fujinami et al. \cite{FLMP25B}, which proceeds to employ a sequence of pointwise bounds to establish gradient dominance. Since all inequalities involved are pointwise, the analysis in \cite{FLMP25B} generalizes due to monotonicity of summation/integration. We jump to the crucial part of the analysis, sketching in the key points, remarking that all intermediate steps proceed in exactly the same way.
    
    \begin{lemma}
    Consider the domain-randomized LQR problem \eqref{eq:dr_lqr} under the heterogeneity,
    $$
    \diam(\Theta) \leq \inf_{\theta \in \Theta} \frac{1}{50000 \displaystyle\sup_{B\in \Theta_B}(\Vert B \Vert,1) J(K_{\mathrm{LQR}}(\theta),\theta)^6 }
    $$
    within the set,
    $$\mathcal{S} \coloneqq \left\{ \begin{array}{c|c} K & J_{DR}(K)\leq 8J_{DR}(K^\star),\quad 
    \Vert \nabla  J_{DR}(K) \Vert_F  \leq \displaystyle\frac{1}{\smash{2^8  \displaystyle\sup_{B\in\Theta}(\Vert B \Vert,1)J_{DR}(K)^3}} \end{array}\right\}
    $$
    we have that $J_{\mathrm{DR}}(K)$ is gradient dominated.
    \begin{proof}
    We simply prove an analogous version of Fujinami et al. \cite{FLMP25B} Lemma III.7 which implies the desired result.
    \begin{align*}
    \Vert R(K,K_{\mathrm{DR}}^*) \Vert_F &= \mathbb{E}\Big[ E(K,\theta)\Sigma(K,\theta)\Big( \big(\Sigma(K,\theta)\big)^{-1}\big( \Sigma(K_{SA}^*,\theta) - \Sigma(K,\theta)\big) \Big) \Big]\\
    &\leq \mathbb{E}\Big[ \Vert \nabla J(K,\theta) \Vert_F \big\Vert\Sigma(K,\theta) ^{-1}\big( \Sigma(K_{\mathrm{DR}}^*,\theta) - \Sigma(K,\theta)\big) \big\Vert_{\op} \Big].
    \end{align*}
    Tackling the $\Sigma(K_{\mathrm{DR}}^*, \theta) - \Sigma(K,\theta)$ term, we see that:
    \begin{align*}
            \Sigma(K_{\mathrm{DR}}^*,\theta) - \Sigma(K,\theta) &= \mathrm{dlyap}\big(A-BK, B(K^*_{\mathrm{DR}} - K)\Sigma(K^*_{\mathrm{DR}}, \theta)\big(A - BK\big)^\top +  \big(A - BK \big)\Sigma(K^*_{\mathrm{DR}},\theta)\big(K_{\mathrm{DR}}^* - K \big)^\top B^\top \\
            &\qquad\qquad\qquad\qquad+  B(K - K^*_{\mathrm{DR}})\Sigma(K^*_{\mathrm{DR}},\theta)(K - K^*_{\mathrm{DR}})^\top B^\top \big)
    \end{align*}
    which implies that,
    \begin{align*}
    &\sqrt{\Vert \Sigma(K,\theta)^{-1}(\Sigma(K_{\mathrm{DR}}^*,\theta) - \Sigma(K,\theta))^\top\Sigma(K_{\mathrm{DR}}^*,\theta) - \Sigma(K,\theta))\Sigma(K,\theta)^{-1} \Vert_{\mathrm{op}}} \leq\\
    &\qquad\qquad\Big(\Vert \Sigma(K^*_{\mathrm{DR}},\theta)\Vert_{\op} \Vert K - K_{\mathrm{DR}}^\star\Vert_{\op}^2 \Vert B\Vert_{\op}^2 + 2 \Vert \Sigma(K^*_{\mathrm{DR}},\theta)\Vert_{\op} \Vert A - BK\Vert_{\op}\Vert K - K_{\mathrm{DR}}^\star\Vert_{\op} \Vert B \Vert_{\op} \Big)
    \end{align*}
    and leads to,
    \begin{align*}
    \Vert R(K,K_{\mathrm{DR}}^*) \Vert_F  &\leq \mathbb{E}\Big[ \big\Vert\nabla_K  J(K,\theta)\big\Vert_F  \big( \Vert K - K_{\mathrm{DR}}^*\Vert_{\op}^2 \tau_B^2 + 2 \Vert K - K_{\mathrm{DR}}^*\Vert_{\op} \tau_B \big)J(K,\theta)^{3/2}\Big]\\
    &\leq s^{3/2}J_{\mathrm{DR}}(K)^{3/2}\big( \Vert K - K_{\mathrm{DR}}^*\Vert_{\op}^2 \tau_B^2 + 2 \Vert K - K_{\mathrm{DR}}^*\Vert_{\op} \tau_B \big)\mathbb{E}\Big[ \big\Vert\nabla_K  J(K,\theta)\big\Vert_F  \Big] \\
    &= s^{3/2}J_{\mathrm{DR}}(K)^{3/2}\big( \Vert K - K_{\mathrm{DR}}^*\Vert_{\op} \tau_B + 1\big)\mathbb{E}\Big[ \big\Vert\nabla_K  J(K,\theta)\big\Vert_F  \Big]\Vert K - K_{\mathrm{DR}}^*\Vert_{\op}.
    \end{align*}
    Now, we can bound $K - K_{\mathrm{DR}}^*$ by,
    \begin{align*}
    \Vert K - K_{\mathrm{DR}}^*\Vert_{\op} &\leq \mathbb{E}[\Vert K_{\mathrm{DR}}^* - K^*(\theta)\Vert_{\op}] + \mathbb{E}[\Vert K - K^*(\theta)\Vert_{\op}]\\
    &\leq \sqrt{2 J_{\mathrm{DR}}(K)}\mathbb{E}\big[\big(\Vert \nabla J(K,\theta)\Vert_{\op} + \Vert \nabla J(K_{\mathrm{DR}}^*,\theta) \Vert_{\op} \big)\big].
    \end{align*}
    From here, we see that our heterogeneity and gradient norm assumptions, we obtain that:
    \begin{align*}
    \big\Vert \nabla_K J(K,\theta)\big\Vert \leq \frac{1}{2}\Big(\frac{1}{2^8 J_{\mathrm{DR}}(K)^{5/2}} + \frac{8 s^{13/2}J_{\mathrm{DR}}(K)^{9/2}}{50000 J_{\mathrm{DR}}(K)^6} + \frac{2s^4J_{\mathrm{DR}}(K)^4}{50000^2 J_{\mathrm{DR}}(K)^{12}} \Big) \leq \frac{1}{16 J_{\mathrm{DR}}(K)^{3/2}}.
    \end{align*}
    The only detail that matters is the exponent on $J_{\mathrm{DR}}(K)$ is 3/2, which cancels out $J_{\mathrm{DR}}(K)$ and that denominator is large enough to yield:
    $$
    \Vert R(K,K_{\mathrm{DR}}^*) \Vert_F  \leq \frac{1}{4} \Vert K - K_{\mathrm{DR}}^\star \Vert_F.
    $$
    which implies the desired result.
    \end{proof}
    \end{lemma}
    \subsection{Discount Annealing}
    \noindent To initialize within $\mathcal{S}$, we simply need to employ the discount annealing scheme introduced in Perdormo et al. \cite{PUS21} as discussed in Fujinami et al. \cite{FLMP25A}. To adapt the technique in Fujinami et al., we will employ the notation:
    $$
    \gamma \theta = [\gamma A(\theta), \gamma B(\theta)], \qquad J_{\mathrm{DR}}(K,\gamma) = \mathbb{E}[J(K,\gamma \theta)], \qquad K^*(\gamma) = \arg\min_{K} J_{\mathrm{DR}}(K,\gamma),
    $$
    and adapt the theorems as needed.
    \begin{proposition}[Pedormo et al.  Thm. 1 Adaptation]\label{prop:discount_growth_factor}
    Let $P(K,\gamma\theta)$ satisfy,
    $$
    P(K,\gamma\theta) = Q + SK + K^\top S^\top + KRK^\top + \gamma(A(\theta)-B(\theta)K)^\top P(K,\gamma\theta)(A(\theta)-B(\theta)K)
    $$
    If $ Q + SK + K^\top S^\top + KRK^\top \succeq I$, for,
    $$
    \tilde{\gamma} = \Big(\frac{1}{8 \Vert P(K,\gamma \theta)\Vert_{\mathrm{op}}^4}+ 1\Big)^2 \gamma
    $$
    we have that,
    $$
    \mathrm{tr}( P(K,\tilde{\gamma}\theta )) \leq 2\mathrm{tr}( P(K,\gamma \theta)) 
    $$
\end{proposition}
We simply remark that in \cite{FLMP25B} and \cite{PUS21}, the specific choice, $S=0$ is made, but the theorem as established in \cite{PUS21} accommodates $S \neq 0$ when $\mathcal{Q}\succeq I$ because,
$$
\begin{bmatrix}I\\-K\end{bmatrix}^\top \begin{bmatrix}Q - I& S\\S^\top &  R - I\end{bmatrix}\begin{bmatrix}I\\-K\end{bmatrix} = (Q-I) - SK - K^\top S^\top + K^\top (R-I) K \succeq 0 \; \implies \; Q-SK-K^\top S^\top + R \succeq I + K^\top R K.
$$
as required by Perdomo et al. ~\cite[Thm. 1]{PUS21}. Thus, we can extend the discounting annealing scheme to many systems (with $S\neq 0$) by simply adapting Fujinami et al. \cite[Theorem IV.1]{FLMP25B}.
\begin{algorithm}
    \caption{Discount Annealing, DR Adaptation} 
    \label{alg:progressive_discounting}
    \begin{algorithmic}[1]
    \STATE \textbf{Input:} Distribution of systems to sample $\theta_i \overset{\text{i.i.d.}}{\sim} p_{\Theta}$, optimization tolerance $\varepsilon$
    \STATE $K \gets 0$
    \STATE Find the largest $\gamma \in (0, 1]$ satisfying 
    \begin{equation}\label{eq:initial_gamma}
        J_{\mathrm{DR}}(K, \gamma) \leq 8 d_x. 
    \end{equation} 
    \WHILE {$\gamma < 1$}
    \STATE Set $K \gets K'$, where $K'$ satisfies,
    \begin{equation}   
        \label{eq:controller_update}
        J_{\mathrm{DR}}(K', \gamma) - \inf_{\tilde K} J_{\mathrm{DR}}(\tilde K ,\gamma) \leq d_x .
    \end{equation}
    \STATE Find a discount factor $\gamma'\in[\gamma, 1]$ satisfying
    \begin{equation}
        2J_{\mathrm{DR}}(K , \gamma) < J_{\mathrm{DR}}(K , \gamma') \leq 4J_{\mathrm{DR}}(K , \gamma) \label{eq:gamma_search}
    \end{equation}
    \STATE Update the discount factor $\gamma \gets \gamma'$
    \ENDWHILE
    \STATE Run \Cref{alg:dr-mba} starting from $K$ to obtain $K'$ satisfying,
    $$J_{\mathrm{DR}}(K') - \inf_{\tilde K} J_{\mathrm{DR}}(\tilde K) \leq  \varepsilon.$$
    \STATE \textbf{Return} $K'$
    \end{algorithmic}
\end{algorithm}
\begin{theorem}[\cite{FLMP25B} Theorem IV.1 Adaptation]
Let $\gamma_0$ be the largest $\gamma$ satisfying \eqref{eq:initial_gamma}. At the start of any iteration, it holds that $J_{\mathrm{DR}}(K, \gamma) \leq 8 \inf_{\tilde K} J_{SA}(K, \gamma)$. Further, \Cref{alg:progressive_discounting} terminates after at most $1024 J_{\mathrm{DR}}(K^\star)^4\log(1/\gamma_0)$ iterations where the controller $K'$ returned jointly stabilizes all systems in $\Theta$, satisfying $J_{\mathrm{DR}}(K') - J_{\mathrm{DR}}(K^*) \leq \varepsilon$. 
\begin{proof} 
    When $\gamma \in (0,1]$ is picked so that $J_{\mathrm{DR}}(K,\gamma) \leq 8d_x$, we see that
    $$
    J_{\mathrm{DR}}(K,\gamma) \leq 8d_x \leq 8 \mathrm{tr}(I)\leq 8\inf_{\tilde{K}} J_{\mathrm{DR}}(\tilde{K},\gamma),
    $$
    so we have the desired inequality for the first iteration within \cref{eq:initial_gamma}. Next, we find a controller where,
$$
J_{\mathrm{DR}}(K^\prime,\gamma) - \inf_{\tilde{K}} J_{\mathrm{DR}}(\tilde{K},\gamma) \leq d_x = \mathrm{tr}(I).
$$
We then use this controller $K^\prime$ to update the discount factor $\gamma^\prime$ where we ensure that $\gamma^\prime \in [\gamma,1]$:
$$
2 J_{\mathrm{DR}}(K^\prime,\gamma) < J_{\mathrm{DR}}(K^\prime,\gamma^\prime)\leq 4 J_{\mathrm{DR}}(K^\prime,\gamma)
$$
then, $J_{\mathrm{DR}}(K,\gamma) \leq 8 \inf_{\tilde{K}} J_{\mathrm{DR}}(K,\gamma)$ because $J_{\mathrm{DR}}(K^\prime,\gamma)$ is within a factor of 2 of the optimal cost,
\begin{equation}\label{eq:update_inequality}
J_{\mathrm{DR}}(K^\prime,\gamma)  \leq  \mathrm{tr}(I) + \inf_{\tilde{K}} J_{\mathrm{DR}}(\tilde{K},\gamma)  \leq 2\inf_{\tilde{K}} J_{\mathrm{DR}}(\tilde{K},\gamma) \tag{\text{Update Inequality}}
\end{equation}
(by the definition of the cost-to-go), and also,
$$
2 J_{\mathrm{DR}}(K^\prime,\gamma)] < J_{\mathrm{DR}}(K^\prime,\gamma^\prime) \leq 4J_{\mathrm{DR}}(K^\prime,\gamma) \overbrace{\leq 8 \inf_{\tilde{K}} J_{\mathrm{DR}}(\tilde{K},\gamma)}^{\mathclap{\text{Update Inequality}}}  \overbrace{\leq  8 \inf_{\tilde{K}} J_{\mathrm{DR}}(\tilde{K},\gamma^\prime) }^{\mathclap{\text{Cost is nondecreasing in $\gamma$}}}.
$$
This establishes that at the start of each iteration, it holds that $J_{\mathrm{DR}}(K, \gamma) \leq 8 \inf_{\tilde K} J_{SA}(K, \gamma)$. Now, we establish that this procedure will terminate. For this, we appeal to \cref{prop:discount_growth_factor} to analyze the minimum discount factor growth by employing a contradiction; assume for a contradiction that,
$$
\gamma^\prime < \inf_{\theta \in \Theta}\Big(\frac{1}{8 \Vert P(K,\gamma,\theta)\Vert_{\mathrm{op}}^4}+ 1\Big)^2 \gamma.
$$
By \cref{prop:discount_growth_factor} and monotonicity of the cost in the discount factor, we must have that for each cost term,
$$
J(K,\gamma^\prime \theta) \leq 2J(K,\gamma \theta) \qquad \implies \qquad J_{\mathrm{DR}}(K^\prime,\gamma^\prime) \leq 2J_{\mathrm{DR}}(K^\prime, \gamma)
$$
However, Step 6 in \Cref{alg:progressive_discounting} forces the value of $\gamma^\prime$ found to ensure that $J_{\mathrm{DR}}(K^\prime,\gamma^\prime) > 2 J_{\mathrm{DR}}(K^\prime, \gamma)$. This is a contradiction, so we must have that,
$$
\gamma^\prime \geq \inf_{\theta \in \Theta}\Big(\frac{1}{8 \Vert P(K,\gamma,\theta)\Vert_{\mathrm{op}}^4}+ 1\Big)^2 \gamma.
$$
Next, we employ a sequence of inequalities to find the minimum growth, where we see that,
$$
    \Vert P(K,\theta,\gamma) \Vert_{\mathrm{op}}^4 \overbrace{\leq \mathrm{tr}(P(K,\theta,\gamma))^4}^{\text{Trace Norm Inequality}} \overbrace{\leq 2^4J_{\mathrm{DR}}(K^*(\gamma),\gamma))}^{\text{$(\mathbf{B},s)$-boundedness}}\overbrace{\leq 2^8J_{\mathrm{DR}}(K^*(\gamma),\gamma^\prime)^4}^\text{Update Inequality} \overbrace{\leq 2^8J_{\mathrm{DR}}(K^*_{\mathrm{DR}},\gamma^\prime)^4}^{\text{Definition of $K^*(\gamma)$}} \overbrace{\leq  2^8J_{\mathrm{DR}}(K^*_{\mathrm{DR}})^4}^{\text{Monotonicity in $\gamma$}}.
$$
From here, we can take the reciprocal (and reverse the inequality directions) to obtain that:
$$
\frac{1}{\Vert P(K,\theta_i,\gamma) \Vert_{\mathrm{op}}^4} \geq \frac{1}{\mathrm{tr}(P(K,\theta_i,\gamma))^4} \geq \frac{1}{2^4J_{\mathrm{DR}}(K^\star(\gamma),\gamma)^4} \geq \frac{1}{2^8J_{\mathrm{DR}}(K^\star)^4}
$$
and therefore,
$$
\left(\frac{1}{8\Vert P(K,\theta,\gamma) \Vert_{\mathrm{op}}^4} + 1 \right)^2\geq \left(\frac{1}{8\mathrm{tr}(P(K,\theta,\gamma))^4}+ 1 \right)^2 \geq \left(\frac{1}{8\cdot2^4J_{\mathrm{DR}}(K,\gamma)^4}+ 1 \right)^2 \geq \left(\frac{1}{8\cdot2^8J_{\mathrm{DR}}(K^\star)^4}+ 1\right)^2.
$$
Inverting the minimum discount-factor growth expression (the right-most term) yields the result.
\end{proof}
\end{theorem}
\subsection{Lipschitz, and Gradient Norm/Variance Bounds}
Thematically, many of the relationships we derive simply come from vectorizing the computations done in \cite{FGKM18} and across the literature.  results in making substitutions of $\mathcal{Q} \leftarrow R$ and $\mathcal{Q} \leftarrow Q$.
\begin{lemma}[Useful Inequalities]
\begin{align*}
    \Vert P(K,\theta) \Vert_{\mathrm{op}} &\leq \frac{J(K,\theta)}{\sigma_{\mathrm{min}}(\Sigma_w
    )}\\
    \Vert \Sigma(K,\theta) \Vert_{\mathrm{op}} &\leq \frac{J(K,\theta)}{\sigma_{\mathrm{min}}\left(\mathcal{Q}\right) }\\
    \big\Vert \begin{bmatrix}I & K^\top\end{bmatrix} \big\Vert_{\mathrm{op}}^2&\leq \frac{J(K,\theta)}{\sigma_{\mathrm{min}}\left(\mathcal{Q}\right) \sigma_{\mathrm{min}}(\Sigma_w
    )}\\
    \Vert E(K,\theta) \Vert_{\mathrm{op}} &\leq \left(\left\Vert\mathcal{Q}\right\Vert_{\mathrm{op}} + \big\Vert\begin{bmatrix}A & B\end{bmatrix}\big\Vert_{\mathrm{op}}^2\Vert P(K,\theta) \Vert_{\mathrm{op}} \right)\left\Vert\begin{bmatrix}I\\-K\end{bmatrix}\right\Vert_{\mathrm{op}}\\
    \Vert R + B^*P(K)B\Vert_{\mathrm{op}} &\leq \left(\left\Vert\mathcal{Q}\right\Vert_{\mathrm{op}} + \big\Vert\begin{bmatrix}A & B\end{bmatrix}\big\Vert_{\mathrm{op}}^2\Vert P(K,\theta) \Vert_{\mathrm{op}} \right)\\
    \Vert B \Delta \Sigma(K,\theta)(A-BK)\Vert_{\mathrm{F}} &\leq \big\Vert\begin{bmatrix}A & B\end{bmatrix} \big\Vert_{\mathrm{op}}^2 \Vert\Sigma(K,\theta)\Vert_{\mathrm{op}}\big\Vert\begin{bmatrix}I & -K^\top\end{bmatrix} \big\Vert_{\mathrm{op}} \Vert \Delta \Vert_F\\
    \Vert\mathrm{dlyap}(A-BK,\cdot)\Vert_{\mathrm{op}} &\leq \frac{J(K,\theta)}{\sigma_{\mathrm{min}}(\Sigma_w)\sigma_{\mathrm{min}}\left(\mathcal{Q}\right)}
\end{align*}
\begin{proof}
The first three inequalities follow by the singular-value rearrangement inequality and the fact that $A,B\succeq 0$.
$$
\mathrm{tr}(AB) \geq \sigma_{\mathrm{max}}(A)\sigma_{\mathrm{min}}(B).
$$
Applying this to,
$$
J(K,\theta) = \mathrm{tr}(P(K,\theta),\Sigma_w) = \mathrm{tr}\left(\begin{bmatrix}I\\K\end{bmatrix}^\top \begin{bmatrix}Q & S\\S^\top & R\end{bmatrix}\begin{bmatrix}I\\K\end{bmatrix}\Sigma(K,\theta)\right) = \mathrm{tr}\left(\begin{bmatrix}I\\K\end{bmatrix}^\top \mathcal{Q}\begin{bmatrix}I\\K\end{bmatrix}\Sigma(K,\theta)\right)
$$
yields the first three results. The second three results follow by the block-matrix identites,
    \begin{align*}
    E(K,\theta) &= -\begin{bmatrix}0 & I\end{bmatrix}\left( \mathcal{Q} + \begin{bmatrix}A^\top\\B^\top\end{bmatrix}P(K)\begin{bmatrix}A&B
    \end{bmatrix}\right)\begin{bmatrix}I\\-K\end{bmatrix}\\
    R + B^\top P(K,\theta)B &= \begin{bmatrix}0 & I\end{bmatrix}\left(\mathcal{Q}+ \begin{bmatrix}A^\top \\B^\top\end{bmatrix}P(K)\begin{bmatrix}A&B \end{bmatrix}\right)\begin{bmatrix}0\\I\end{bmatrix}\\
    B \Delta \Sigma(K,\theta)(A-BK) &= \begin{bmatrix}A & B\end{bmatrix}\begin{bmatrix}0\\ I\end{bmatrix}\Delta \Sigma(K) \begin{bmatrix}I & -K^\top\end{bmatrix}\begin{bmatrix}
        A^\top\\ B^\top
    \end{bmatrix}
    \end{align*}
    followed by judicious application of operator norm bounds. The very last inequality follows by the exact same argument as in Fazel et al. \cite{FGKM18}
\end{proof}
\end{lemma}
\noindent These inequalities, along with \cref{lem:bs_boundedness}, suffice to for finding bounds on the Lipschitz constant. Once again, we also introduce,
\begin{equation}\label{eq:largest_system}
    \bar{\theta} = \sup_{\theta \in \Theta} \Vert \theta \Vert
\end{equation}
to simplify the notation.
\begin{corollary}[Lipschitz Bound] \label{cor:hessian_lipschitz_characterization}
    $$
    L_K \leq 4\left(\left\Vert\mathcal{Q}\right\Vert_{\mathrm{op}} + \frac{2\big\Vert\bar{\theta} \big\Vert_{\mathrm{op}}^2 J_{DR}(K_0)}{\sigma_{\mathrm{min}}(\Sigma_w)} \right)  \left(1 + \frac{4\Vert \bar{\theta}\big\Vert_{\mathrm{op}}^2J_{DR}(K_0)}{\sigma_{\mathrm{min}}(\Sigma_w)^2} \right)\frac{2J_{DR}(K_0)}{\sigma_{\mathrm{min}}(\mathcal{Q}) }
    $$
\end{corollary}
\begin{lemma}[Useful Derivatives]
    \begin{align*}
        D_{\theta}E(K,\theta)[\Delta_\theta] &= -\begin{bmatrix}0 & I\end{bmatrix}\left( \begin{bmatrix}\Delta_A^\top\\\Delta_B^\top\end{bmatrix}P(K,\theta)\begin{bmatrix}A&B
    \end{bmatrix}\right)\begin{bmatrix}I\\-K\end{bmatrix}\\
    &\qquad + \begin{bmatrix}0 & I\end{bmatrix}\left(  \begin{bmatrix}A^\top\\B^\top\end{bmatrix}D_{\theta}P(K,\theta)[\Delta_\theta]\begin{bmatrix}A&B
    \end{bmatrix}\right)\begin{bmatrix}I\\-K\end{bmatrix}\\
    &\qquad + \begin{bmatrix}0 & I\end{bmatrix}\left( \begin{bmatrix} A^\top\\B^\top\end{bmatrix}P(K,\theta)\begin{bmatrix}\Delta_A&\Delta_B
    \end{bmatrix}\right)\begin{bmatrix}I\\-K\end{bmatrix}\\
    D_{\theta}P(K,\theta)[\Delta_\theta] &= \mathrm{dlyap}\Big((A-BK)^*, \begin{bmatrix}I\\-K\end{bmatrix}^*\begin{bmatrix}A & B\end{bmatrix}^*P(K,\theta)\begin{bmatrix}\Delta_A & \Delta_B\end{bmatrix}\begin{bmatrix}I\\-K\end{bmatrix}\Big)\\
    &\qquad + \mathrm{dlyap}\Big((A-BK)^*,\begin{bmatrix}I\\-K\end{bmatrix}^*\begin{bmatrix}\Delta_A & \Delta_B\end{bmatrix}^*P(K,\theta)\begin{bmatrix} A & B\end{bmatrix}\begin{bmatrix}I\\-K\end{bmatrix}\Big)\\
    D_{\theta} \Sigma(K,\theta)[\Delta_\theta] &= \mathrm{dlyap}\left(A-BK, \begin{bmatrix}A & B\end{bmatrix}\begin{bmatrix}I\\-K\end{bmatrix}\Sigma(K,\theta)\begin{bmatrix}I\\-K\end{bmatrix}^{\top}\begin{bmatrix}\Delta_A & \Delta_B\end{bmatrix}^{\top} \right)\\
    &\qquad +\mathrm{dlyap}\left(A-BK, \begin{bmatrix}\Delta_A & \Delta_B\end{bmatrix}\begin{bmatrix}I\\-K\end{bmatrix}\Sigma(K,\theta)\begin{bmatrix}I\\-K\end{bmatrix}^{\top}\begin{bmatrix}A & B\end{bmatrix}^{\top}\right)
    \end{align*}
\end{lemma}
\noindent Similarly, we can find bounds on the operator norm bound, which are used twice to obtain \cref{lem:matrix_bernstein}.
\begin{lemma}[Lipschitz System Perturbation Bound]\label{lem:var_poly_bound}
    \begin{align*}
    &\Vert D_\theta \Big(2E(K,\theta) \Sigma(K,\theta)\Big)[\Delta_{\theta}]\Vert_{\mathrm{op}} = \Vert2 D_\theta E(K,\theta)[\Delta_\theta] \Sigma(K,\theta) + 2 E(K,\theta) D_{\theta} \Sigma(K,\theta)[\Delta_\theta]\Vert_{\mathrm{op}}\\
    &\qquad \leq 4\left(\frac{2^{5/2}J_{\mathrm{DR}}(K_0)^{5/2} \big\Vert \bar{\theta}\big\Vert_{\mathrm{op}}}{\sigma_{\mathrm{min}}(\Sigma_w)^{3/2}\sigma_{\mathrm{min}}(\mathcal{Q})^{3/2} } + \frac{2^{7/2}J_{\mathrm{DR}}(K_0)^{7/2}\big\Vert\bar{\theta}\big\Vert_{\mathrm{op}}\Vert \mathcal{Q}\Vert_{\mathrm{op}} }{\sigma_{\mathrm{min}}(\Sigma_w)^{5/2}\sigma_{\mathrm{min}}(\mathcal{Q})^{7/2}} + \frac{2^{11/2}J_{\mathrm{DR}}(K)^{9/2}\big\Vert\bar{\theta}\big\Vert^3_{\mathrm{op}}}{\sigma_{\mathrm{min}}(\Sigma_w)^{7/2}\sigma_{\mathrm{min}}(\mathcal{Q})^{7/2}}\right)
    \end{align*}
    \begin{proof} We bound the expression termwise using the triangle inequality. The first term can be bounded above by,
        \begin{align*}
            &\Vert D_{\theta} E(K,\theta)[\Delta_\theta] \Sigma(K,\theta)\Vert_F \leq 2 \big\Vert P(K,\theta)\big\Vert_{\mathrm{op}}\big\Vert \! \begin{bmatrix}A & B\end{bmatrix}\!\big\Vert_{\mathrm{op}}\big\Vert \! \begin{bmatrix}I & K^\top\end{bmatrix}\!\big\Vert_{\mathrm{op}}\Vert \Sigma(K,\theta)\Vert_{\mathrm{op}}\Vert \Delta_{\theta}\Vert_F\\
            &\qquad + 2 \big\Vert \! \begin{bmatrix}A & B\end{bmatrix}\!\big\Vert_{\mathrm{op}}^3\big\Vert \! \begin{bmatrix}I & K^\top\end{bmatrix}\!\big\Vert_{\mathrm{op}}^3\Vert \Sigma(K,\theta)\Vert_{\mathrm{op}} \Vert P(K,\theta)\Vert_{\mathrm{op}}  \big\Vert \mathrm{dlyap}(A-BK,\cdot) \Vert_{\mathrm{op}} \Vert\Delta_{\theta}\Vert_F\\
            &\phantom{\Vert D_{\theta} E(K,\theta)[\Delta_\theta] \Sigma(K,\theta)\Vert_F} \leq \left(2 \frac{J(K,\theta)^{5/2} \big\Vert \!\begin{bmatrix}
            A & B\end{bmatrix} \!\big\Vert_{\mathrm{op}}}{\sigma_{\mathrm{min}}(\Sigma_w)^{3/2}\sigma_{\mathrm{min}}(\mathcal{Q})^{3/2} } + 2 \frac{J(K,\theta)^{9/2}\big\Vert\!\begin{bmatrix}A & B\end{bmatrix}\!\big\Vert^3_{\mathrm{op}}}{\sigma_{\mathrm{min}}(\Sigma_w)^{7/2}\sigma_{\mathrm{min}}(\mathcal{Q})^{7/2}}\right)\Vert \Delta_\theta\Vert_{\mathrm{F}}.
        \end{align*}
        The second term can be bounded above by,
        \begin{align*}
            &\Vert E(K,\theta) D_{\theta} \Sigma(K,\theta)[\Delta_\theta]\Vert_{\mathrm{F}} \leq \\
            &\qquad\leq 2(\Vert \mathcal{Q}\Vert_{\mathrm{op}} + \big\Vert \!\begin{bmatrix}A & B\end{bmatrix} \!\big\Vert_{\mathrm{op}}^2\Vert P(K,\theta) \Vert_{\mathrm{op}})\big\Vert \!\begin{bmatrix}I & -K^\top\end{bmatrix}\!\big\Vert_{\mathrm{op}}^3\big\Vert\mathrm{dlyap}(A-BK,\cdot)\big\Vert_{\mathrm{op}}\big\Vert\!\begin{bmatrix}A & B\end{bmatrix}\!\big\Vert_{\mathrm{op}}\Vert \Sigma(K,\theta)\Vert_{\mathrm{op}}\Vert \Delta_\theta\Vert_{\mathrm{F}}\\
            &\qquad\leq 2\left(\Vert \mathcal{Q}\Vert_{\mathrm{op}} + \big\Vert \!\begin{bmatrix}A & B\end{bmatrix} \!\big\Vert_{\mathrm{op}}^2\frac{J(K,\theta)}{\sigma_{\mathrm{min}}(\Sigma_w) }\right)\frac{J(K,\theta)^{7/2}}{\sigma_{\mathrm{min}}(\Sigma_w)^{5/2}\sigma_{\mathrm{min}}(\mathcal{Q})^{7/2}}\big\Vert\!\begin{bmatrix}A & B\end{bmatrix}\!\big\Vert_{\mathrm{op}}\Vert \Delta_\theta\Vert_{\mathrm{F}}.
            \end{align*}
            Combining the two yields, and assuming $\Vert \Delta_\theta\Vert_{\mathrm{F}}=1$ immediately yields the operator norm bound:
            \begin{align*}
                \Vert D_\theta \Big(2E(K,\theta) \Sigma(K,\theta)\Big)[\Delta_{\theta}]\Vert_{\mathrm{op}} &\leq 4\left(\frac{J(K,\theta)^{5/2} \big\Vert \!\begin{bmatrix}
                A & B\end{bmatrix} \!\big\Vert_{\mathrm{op}}}{\sigma_{\mathrm{min}}(\Sigma_w)^{3/2}\sigma_{\mathrm{min}}(\mathcal{Q})^{3/2} } + \frac{J(K,\theta)^{9/2}\big\Vert\!\begin{bmatrix}A & B\end{bmatrix}\!\big\Vert^3_{\mathrm{op}}}{\sigma_{\mathrm{min}}(\Sigma_w)^{7/2}\sigma_{\mathrm{min}}(\mathcal{Q})^{7/2}}\right)\\
                &\qquad + 4\left(\Vert \mathcal{Q}\Vert_{\mathrm{op}} + \big\Vert \!\begin{bmatrix}A & B\end{bmatrix} \!\big\Vert_{\mathrm{op}}^2\frac{J(K,\theta)}{\sigma_{\mathrm{min}}(\Sigma_w) }\right)\frac{J(K,\theta)^{7/2}\big\Vert\!\begin{bmatrix}A & B\end{bmatrix}\!\big\Vert_{\mathrm{op}} }{\sigma_{\mathrm{min}}(\Sigma_w)^{5/2}\sigma_{\mathrm{min}}(\mathcal{Q})^{7/2}}
            \end{align*}
            A minor simplification yields the result.
    \end{proof}
\end{lemma}
\subsection{Matrix Bernstein Bounds}
The bounds found in the previous section are with respect to the Frobenius norm, while the Matrix Bernstein concentration inequalities are with respect to the operator norm. These discrepancies in norm formulations occur frequently, and Fazel et al. \cite{FGKM18} simply use vector Bernstein bounds to tackle concentration since $\Vert X \Vert_{\op} \leq \Vert X \Vert_{\mathrm{F}} = \Vert \vect(X) \Vert_{2}$. For the sake of completeness, we demonstrate how to derive the bound given in the main text.
\begin{lemma}[Matrix Bernstein Inequality] Let $G_i$ be defined as in \eqref{eq:sym_bernstein}, and define,
$$
\begin{gathered}
\Vert G_i\Vert_{\mathrm{F}} \leq \bar{G}, \qquad \mathbb{E}[\Vert G_i \Vert_F^2] \leq \sigma^2
\end{gathered}
$$
Then\footnote{We remark that compactness of sublevel-sets and continuity of $J_{\text{DR}}$ ensure that $\bar{G}$ and $\sigma^2$ are finite.},
\begin{align*}
\Prob\left\{ \Vert G_M\Vert_{F} \geq \varepsilon_{\nabla}  \right\} \!\leq \!2(n_xn_u\!+\!1)\!\exp\!\left(\!\frac{-M \varepsilon_\nabla^2}{4 (\sqrt{2}\sigma^2 \!+\! \bar{G}\varepsilon_\nabla/3) }\!\right)\!,
\end{align*}
so if we pick $M$ as,
$$
M \geq \frac{4(\sqrt{2}\sigma^2 + \bar{G}\varepsilon_\nabla/3)}{\varepsilon_{\nabla}^2}\log\left(\frac{2(n_xn_u+1)}{\delta}\right).
$$
with probability at least $1-\delta$, each gradient-estimation error is bounded above by $\varepsilon_\nabla$.
\begin{proof}
First, we transform the Frobenius norm into an operator norm so that we may use the rectangular matrix bernstein bound found in Vershynin \cite{V18} Exercise 5.4.15.
\begin{align*}
\Prob\left(\Vert g(K) - \nabla J_{\mathrm{DR}}(K) \Vert_{\mathrm{F}} \geq \varepsilon_\nabla\right) &= \Prob\left(\sqrt{2}\Vert g(K) - \nabla J_{\mathrm{DR}}(K) \Vert_{\mathrm{op}} \geq \varepsilon_\nabla\right)=\Prob\left(\frac{\sqrt{2}}{M}\Big\Vert \sum_{i=1}^M G_i \Big\Vert_{\mathrm{op}} \geq \varepsilon_\nabla\right)\\
&=\Prob\left(\Big\Vert \sum_{i=1}^M G_i \Big\Vert_{\mathrm{op}} \geq \frac{M}{\sqrt{2}}\varepsilon_\nabla\right)
\end{align*}
In order to use the matrix Bernstein inequality, we need to bound $\Vert G_i\Vert_{\mathrm{op}}$ and $\mathbb{E}[\Vert G_M^2\Vert_{\mathrm{op}}]$. First, we see that $\Vert G_i\Vert_{\mathrm{op}}$ can be evaluated directly by obtaining its nontrivial eigenvalues, which are $\pm \Vert \nabla_K (J(K,\theta_i) - \mathbb{E}[J(K,\theta)]\Vert_F = \Vert G_i \Vert_F$, and so,
$$
\Vert G_i\Vert_{\mathrm{op}} = \Vert G_i\Vert_{\mathrm{F}} \leq \bar{G}.
$$
Next, to evaluate $\mathbb{E}[\Vert G_M^2\Vert_{\mathrm{op}}]/M$, let $v(\theta_i) = \vect\big(\nabla_K (J(K,\theta_i) - \mathbb{E}J(K,\theta))\big)$. 
\begin{align*}
    \mathbb{E}[G_M^2] 
    &= M\begin{bmatrix} \mathbb{E}[v(\theta)v(\theta)^\top] \\ & \mathbb{E}[v(\theta)^\top v(\theta)]
    \end{bmatrix}
    \end{align*}
where we make judicious use of the i.i.d. assumption on $\theta_i$. To find the operator norm, we observe that:
    \begin{align*}
    \Vert \mathbb{E}[
    G_M^2] \Vert_{\op} &=  M\max\left\{ \Vert \mathbb{E}[v(\theta)v(\theta)^\top]\Vert_{\op}], \mathbb{E}[ v(\theta)^\top v(\theta)]\right\}.
    \end{align*}
    Only the second term matters because,
    \begin{align*}
    \mathbb{E}[v^\top v] &= \mathbb{E}[\tr( v^\top v )] = \mathbb{E}[\tr( v v^\top )] = \tr(\mathbb{E}[vv^\top])= \sigma_1(\mathbb{E}[vv^\top]) + \ldots + \sigma_n(\mathbb{E}[vv^\top]) \geq \Vert \mathbb{E}[vv^\top]\Vert_{\op},
    \end{align*}
    and undoing the vectorization operation gives us that,
    $$
    \mathbb{E}[\Vert G_M^2\Vert_{\mathrm{op}}] = M\mathbb{E}[\Vert G_i \Vert_F^2] \leq M\sigma^2
    $$
    Putting this all together, we see that
    \begin{align*}
    \Prob\left(\Vert g(K) - \nabla J_{\mathrm{DR}}(K) \Vert_{\mathrm{F}} \geq \varepsilon_\nabla\right) &= \Prob\left(\Big\Vert \sum_{i=1}^M G_i \Big\Vert_{\mathrm{op}} \geq \frac{M}{\sqrt{2}}\varepsilon_\nabla\right)\leq 2 (d_xd_u+1)\exp\Big(-\frac{\frac{M^2\varepsilon_\nabla^2}{2}}{2\big(M\sigma^2 + \frac{\bar{G}}{3}\frac{M \varepsilon_{\nabla}}{\sqrt{2}} \big)} \Big)\\
    &= 2 (d_xd_u+1)\exp\Big(-\frac{3\sqrt{2}M\varepsilon_\nabla^2}{4\big(3\sqrt{2}\sigma^2 + \bar{G} \varepsilon_{\nabla} \big)} \Big)
    \end{align*}
    Finally, we can fix a failure probability $\delta$, plug in constants, and rearrange to obtain:
    $$
    M\geq\frac{4\big(3\sqrt{2}\,\sigma^2 + \bar{G} \varepsilon_{\nabla} \big)}{3\sqrt{2}\,\varepsilon_{\nabla}^2}\log\left(\frac{2(n_xn_u+1)}{\delta}\right)
    $$
    \end{proof}
\end{lemma}
\subsection{Polynomial Expressions}
\noindent First, we collect, two expressions from lemma 27 in Fazel et. al. \cite{FGKM18} that are used to establish convergence of zeroth-order gradient methods, adapted to the task at hand.
\begin{equation}\label{eq:r_guard}
    r_g = \inf_{\theta \in \Theta}\!\min\!\left\{\!\frac{\sigma_{\mathrm{min}}(\mathcal{Q}) \sigma_{\mathrm{min}}(\Sigma_w)}{4 J(K^{\gd},\theta) \Vert B \Vert_{\op}(\Vert A\!-\!BK^{\gd} \Vert_{\op} \!+\! 1)},\Vert K^{\gd} \Vert_{\op}\!\right\} 
\end{equation}
\begin{equation}\label{eq:gradient_error_to_cost_error_poly}
\begin{aligned}
L_{\mathrm{cost}}&=4  J_{\mathrm{DR}}(K) \Vert \mathcal{Q}\Vert_{\mathrm{op}} \left(2\Vert K \Vert_{\mathrm{op}} + \frac{1}{4 \big\Vert \bar{\theta}\big\Vert_{\mathrm{op}}} \right)\Vert K - K^\prime\Vert_{\mathrm{op}} \nonumber \\
&\quad + \frac{2^3J_{\mathrm{DR}}(K)^{5/2} \big\Vert \bar{\theta} \big\Vert_{\mathrm{op}}}{\sigma_{\mathrm{min}}(\Sigma_w)^{1/2}\sigma_{\mathrm{min}}(\mathcal{Q})^{1/2} }\left(\big\Vert \bar{\theta} \big\Vert_{\mathrm{op}} \frac{2^{1/2}J_{\mathrm{DR}}(K)^{1/2}}{\sigma_{\min}(\mathcal{Q})^{1/2}\sigma_{\min}(\Sigma_w)^{1/2} } +1\right)
\end{aligned}
\end{equation}
Note the dependence on $\Vert K^\mathrm{gd}\Vert_{\mathrm{op}}$, which is detrimental to the task at hand because we employ the discount annealing scheme that starts at $K = 0$. We slightly adapt the analysis from Fazel et al. \cite{FGKM18} to circumvent this issue, and to avoid any dependence of the sampling size on the current controller.
\begin{lemma}
    Let,
    \begin{equation}\label{eq:c_g_definition}
    \Vert K^\prime - K\Vert_{\mathrm{op}} \leq \frac{\sigma_{\mathrm{min}}(Q) \sigma_{\mathrm{min}}(\Sigma_w)}{4 J_{\mathrm{DR}}(K) \big\Vert \bar{\theta}\big\Vert_{\mathrm{op}}\big(\big\Vert \bar{\theta}\big\Vert_{\mathrm{op}} + 1\big)} \eqqcolon c_g
\end{equation}
then,
\begin{align}
    |J_{\mathrm{DR}}(K^\prime) - J_{\mathrm{DR}}(K)|&\leq 4  J_{\mathrm{DR}}(K) \Vert \mathcal{Q}\Vert_{\mathrm{op}} \left(2\Vert K \Vert_{\mathrm{op}} + \frac{1}{4 \big\Vert \bar{\theta}\big\Vert_{\mathrm{op}}} \right)\Vert K - K^\prime\Vert_{\mathrm{op}} \nonumber \\
&\quad + \frac{2^3J_{\mathrm{DR}}(K)^{5/2} \big\Vert \bar{\theta} \big\Vert_{\mathrm{op}}}{\sigma_{\mathrm{min}}(\Sigma_w)^{1/2}\sigma_{\mathrm{min}}(\mathcal{Q})^{1/2} }\left(\big\Vert \bar{\theta} \big\Vert_{\mathrm{op}} \frac{2^{1/2}J_{\mathrm{DR}}(K)^{1/2}}{\sigma_{\min}(\mathcal{Q})^{1/2}\sigma_{\min}(\Sigma_w)^{1/2} } +1\right) \Vert K-K^\prime \Vert_{\mathrm{op}}\nonumber \\
&\eqqcolon L_{\mathrm{cost}} \Vert K - K^\prime \Vert_{\mathrm{op}} \label{eq:L_cost_definition}
\end{align}
\begin{proof}
\begin{align*} &\Vert P(K',\theta) - P(K,\theta)\Vert_{\mathrm{op}} 
= \left\Vert\mathrm{dlyap}\left(A-BK^\prime, \begin{bmatrix}I\\-K^\prime\end{bmatrix}^\top \mathcal{Q} \begin{bmatrix}I\\-K^\prime\end{bmatrix}\right) - \mathrm{dlyap}\left(A-BK, \begin{bmatrix}I\\-K\end{bmatrix}^\top \mathcal{Q} \begin{bmatrix}I\\-K\end{bmatrix}\right)\right\Vert_{\mathrm{op}}\\
&\qquad\qquad \leq \left\Vert\mathrm{dlyap}\left(A-BK^\prime, \begin{bmatrix}I\\-K^\prime\end{bmatrix}^\top \mathcal{Q} \begin{bmatrix}I\\-K^\prime\end{bmatrix}\right) - \mathrm{dlyap}\left(A-BK, \begin{bmatrix}I\\-K^\prime\end{bmatrix}^\top \mathcal{Q} \begin{bmatrix}I\\-K^\prime\end{bmatrix}\right)\right\Vert_{\mathrm{op}} \\
&\qquad\qquad\qquad +  \left\Vert\mathrm{dlyap}\left(A-BK, \begin{bmatrix}I\\-K\end{bmatrix}^\top \mathcal{Q} \begin{bmatrix}I\\-K\end{bmatrix}\right) - \mathrm{dlyap}\left(A-BK, \begin{bmatrix}I\\-K^\prime\end{bmatrix}^\top \mathcal{Q} \begin{bmatrix}I\\-K^\prime\end{bmatrix}\right)\right\Vert_{\mathrm{op}} \\
&\qquad\qquad\leq 2\left\Vert \mathrm{dlyap}(A-BK,\cdot) \right\Vert_{\mathrm{op}} \left\Vert \begin{bmatrix}I\\-K\end{bmatrix}^\top \mathcal{Q} \begin{bmatrix}I\\-K\end{bmatrix}  - \begin{bmatrix}I\\-K^\prime\end{bmatrix}^\top \mathcal{Q} \begin{bmatrix}I\\-K^\prime\end{bmatrix}\right\Vert_{\mathrm{op}} \\
&\qquad\qquad\qquad + 2\left\Vert \mathrm{dlyap}(A-BK,\cdot) \right\Vert_{\mathrm{op}}^2 \Vert (A-BK)(\cdot)(A-BK) - (A-BK^\prime)(\cdot)(A-BK^\prime) \Vert_{\mathrm{op}}\\
&\qquad\qquad\qquad\quad \cdot \left\Vert \begin{bmatrix}I\\-K\end{bmatrix}^\top \mathcal{Q} \begin{bmatrix}I\\-K\end{bmatrix}\right\Vert_{\mathrm{op}} \\
&\qquad\qquad\leq 2\left\Vert \mathrm{dlyap}(A-BK,\cdot) \right\Vert_{\mathrm{op}} (2\Vert K \Vert_{\mathrm{op}} \Vert \mathcal{Q}\Vert_{\mathrm{op}}\Vert K - K^\prime\Vert + \Vert \mathcal{Q}\Vert_{\mathrm{op}} \Vert K - K^\prime\Vert^2) \\
&\qquad\qquad\qquad + 2\left\Vert \mathrm{dlyap}(A-BK,\cdot) \right\Vert_{\mathrm{op}}^2 \big\Vert \begin{bmatrix}A & B\end{bmatrix} \big\Vert_{\mathrm{op}}\left(\big\Vert \begin{bmatrix}A & B\end{bmatrix} \big\Vert_{\mathrm{op}} \left\Vert \begin{bmatrix}I\\-K\end{bmatrix}\right\Vert_{\mathrm{op}}+1\right)\Vert K \Vert_{\mathrm{op}} \Vert K-K^\prime \Vert_{\mathrm{op}}\\
&\qquad\qquad\leq 2 J(K,\theta) \Vert \mathcal{Q}\Vert_{\mathrm{op}} \left(2\Vert K \Vert_{\mathrm{op}} + \frac{\sigma_{\mathrm{min}}(Q) \sigma_{\mathrm{min}}(\Sigma_w)}{4 J(K,\theta) \big\Vert \!\begin{bmatrix}A & B\end{bmatrix} \!\big\Vert_{\mathrm{op}}\big(\big\Vert\! \begin{bmatrix}A & B\end{bmatrix}\! \big\Vert_{\mathrm{op}} + 1\big)}\right)\Vert K - K^\prime\Vert_{\mathrm{op}} \\
&\qquad\qquad\qquad + 2J(K,\theta)^2 \big\Vert \begin{bmatrix}A & B\end{bmatrix} \big\Vert_{\mathrm{op}}\left(\big\Vert \begin{bmatrix}A & B\end{bmatrix} \big\Vert_{\mathrm{op}} \Big(\frac{J(K,\theta)^{1/2}}{\sigma_{\min}(\mathcal{Q})^{1/2}\sigma_{\min}(\Sigma_w)^{1/2} } +1\Big) \right)\Vert K \Vert_{\mathrm{op}} \Vert K-K^\prime \Vert_{\mathrm{op}}\\
&\qquad\qquad\leq 2 J(K,\theta) \Vert \mathcal{Q}\Vert_{\mathrm{op}} \left(2\Vert K \Vert_{\mathrm{op}} + \frac{1}{4 \big\Vert \!\begin{bmatrix}A & B\end{bmatrix} \!\big\Vert_{\mathrm{op}}} \right)\Vert K - K^\prime\Vert_{\mathrm{op}} \\
&\qquad\qquad\qquad + 2J(K,\theta)^2 \big\Vert \begin{bmatrix}A & B\end{bmatrix} \big\Vert_{\mathrm{op}}\left(\big\Vert \begin{bmatrix}A & B\end{bmatrix} \big\Vert_{\mathrm{op}} \Big(\frac{J(K,\theta)^{1/2}}{\sigma_{\min}(\mathcal{Q})^{1/2}\sigma_{\min}(\Sigma_w)^{1/2} } +1\Big) \right)\Vert K \Vert_{\mathrm{op}} \Vert K-K^\prime \Vert_{\mathrm{op}}
\end{align*}
Recall that,
$$
\Vert K \Vert_{\mathrm{op}} \leq \left\Vert\begin{bmatrix}I\\-K\end{bmatrix}\right\Vert_{\mathrm{op}} \leq \frac{J(K,\theta)^{1/2}}{\sigma_{\mathrm{min}}(\Sigma_w)^{1/2}\sigma_{\mathrm{min}}(\mathcal{Q})^{1/2}}
$$
Minor simplifications yields the result. 
\end{proof}
\end{lemma}
\subsection{Additional Experimental Details}
\noindent We use the following dynamical system model for the cart-pole system:
\begin{align*}
\dot{x} &= \dot{x}\\
\ddot{x} & =\frac{m_pg \sin\theta \cos\theta
- \frac{7}{3} \!\left[u + m_p \hat{l} \dot{\theta}^2 \sin\theta - \mu_c \dot{x} \right] - \frac{\mu_p \dot{\theta} \cos\theta}{\hat{l}}}{m_p \cos^2\theta - \frac{7}{3} (m_p+m_c)}\\
\dot{\theta} &= \dot{\theta}\\
\ddot{\theta} &= \frac{3}{7 \hat{l}}
\left( g \sin\theta - \ddot{x} \cos\theta - \frac{\mu_p \dot{\theta}}{m \hat{l}}\right).
\end{align*}
where the sign convention is, a positive sign attached to $x$ corresponds to a right displacement, $\theta = 0$ corresponds to upright position, and a  clockwise rotation for the pendulum is positive. Further, a positive sign attached to $u$, corresponds to rightward push, and this dynamical system is linearized by automatic differentiation about $\boldsymbol{x} = [x,\dot{x}, \theta,\dot{\theta}] = 0$.
\begin{table}[h!]
\centering
\caption{Parameters for the Cartpole System}
\begin{tabular}{lll}
\hline
\textbf{Parameter} & \textbf{Description} & \textbf{Default Value} \\
\hline
$u$ & Force applied (Newtons) & --- \\
$m_c$ & Mass of cart & 1.0 kg\\
$m_p$ & Mass of pole & 0.1 kg\\
$l$ & Length of pole  & 0.2--0.8 m\\
$g$ & Gravitational acceleration  & 9.81 (m/s\textsuperscript{2})\\
$\mu_c$ & Friction coefficient between cart and track & 0.25 \\
$\mu_p$ & Friction coefficient at pendulum joint & 0.01 \\
$\boldsymbol{x}$ & State vector $[x, \dot{x}, \theta, \dot{\theta}]$ & --- \\
\end{tabular}
\end{table}
\newpage

\end{document}